\documentclass{llncs}

\usepackage{algorithm}
\usepackage[noend]{algpseudocode} 
\algrenewcommand\algorithmicindent{0.8em}
\usepackage{wrapfig} 

\newcommand{\LineIf}[2]{%
    \State\algorithmicif\ {#1}\ \algorithmicthen\ {#2}}
\newcommand{\LineIfElse}[3]{%
    \State\algorithmicif\ {#1}\ \algorithmicthen\ {#2}%
    \ \algorithmicelse\ {#3}}

\algdef{SE}[PROCEDURE]{ProcedureRet}{EndProcedureRet}%
[3]{\algorithmicprocedure\ \textproc{#1}\ifthenelse{\equal{#2}{}}{}{(#2)}%
, \textbf{returns}: {#3}}%
{\algorithmicend\ \algorithmicprocedure}%

\algdef{SE}[PROCEDURE]{ProcedureLi}{EndProcedureLi}%
[3]{\algorithmicprocedure\ \textproc{#1}\ifthenelse{\equal{#2}{}}{}{(#2)}%
,\quad\quad \textbf{returns}: {#3}}%
{\algorithmicend\ \algorithmicprocedure}%

\usepackage{graphicx}
\usepackage{amsmath}
\usepackage{amssymb}
\usepackage{stmaryrd}
\usepackage{hyperref}
\usepackage{breakcites}
\usepackage{xspace}
\usepackage{subfigure}
\usepackage{xspace}
\usepackage{enumerate}
\usepackage{booktabs}

\newcommand{\true}{\mathsf{true}}
\newcommand{\false}{\mathsf{false}}
\DeclareMathOperator{\scope}{\mathbin{.}}

\newcommand{\FE}[1]{\mathsf{Force}^a_{#1}}
\newcommand{\FS}[1]{\mathsf{Force}^p_{#1}}

\newcommand{\depqbf}{\textsf{DepQBF}\xspace}
\newcommand{\bloqqer}{\textsf{Bloqqer}\xspace}
\newcommand{\bloqqerM}{\textsf{BloqqerM}\xspace}
\newcommand{\minisat}{\textsf{Minisat}\xspace}
\newcommand{\lingeling}{\textsf{Lingeling}\xspace}
\newcommand{\picosat}{\textsf{PicoSat}\xspace}
\newcommand{\qbfcert}{\textsf{QBFCert}\xspace}
\newcommand{\cudd}{\textsf{CuDD}\xspace}
\newcommand{\iprover}{\textsf{iProver}\xspace}
\newcommand{\Abc}{\textsf{ABC}\xspace}
\newcommand{\vltomv}{\textsf{vl2mv}\xspace}
\newcommand{\ratsy}{\textsf{RATSY}\xspace}
\newcommand{\aiger}{\textsf{AIGER}\xspace}
\newcommand{\propsat}{\textsc{PropSat}}
\newcommand{\propsatmodel}{\textsc{PropSatModel}}
\newcommand{\propsatcore}{\textsc{PropUnsatCore}}
\newcommand{\qbfsat}{\textsc{QbfSat}}
\newcommand{\qbfsatmodel}{\textsc{QbfSatModel}}

\begin{document}
\frontmatter          
\pagestyle{plain}
\mainmatter              
\title{SAT-Based Synthesis Methods for Safety Specs%
\thanks{This work was supported in part by the Austrian Science Fund (FWF)
through projects RiSE (S11406-N23 and S11408-N23) and QUAINT (I774-N23), and
by the European Commission through project STANCE (317753).}}
\titlerunning{SAT-Based Synthesis Methods for Safety Specs}
%
\author{Roderick Bloem$^1$ \and
        Robert K\"onighofer$^1$ \and
        Martina Seidl$^2$
}
\authorrunning{Bloem et al.} 
%
%
\institute{$^1$ Institute for Applied Information Processing and
           Communications (IAIK)\\
           Graz University of Technology, Austria.\\
           \vspace{0.1cm}
           $^2$ Institute for Formal Models and Verification\\
           Johannes Kepler University, Linz, Austria.
           }

\maketitle              

\begin{abstract}
Automatic synthesis of hardware components from declarative specifications is an
ambitious endeavor in computer aided design.  Existing synthesis algorithms are
often implemented with Binary Decision Diagrams (BDDs), inheriting their
scalability limitations. Instead of BDDs, we propose several new methods to
synthesize finite-state systems from safety specifications using decision
procedures for the satisfiability of quantified and unquantified Boolean
formulas (SAT-, QBF- and EPR-solvers). The presented approaches are based on
computational learning, templates, or reduction to first-order logic.  We also
present an efficient parallelization, and optimizations to utilize reachability
information and incremental solving. Finally, we compare all methods in an
extensive case study.  Our new methods outperform BDDs and other existing work
on some classes of benchmarks, and our parallelization achieves a super-linear
speedup.  This is an extended version of \cite{BloemKS14}, featuring an
additional appendix.
\end{abstract}
\section{Introduction}
\label{sec:intro}

Automatic synthesis is an appealing approach to construct correct reactive
systems:  Instead of manually developing a system and verifying it later against
a formal specification, reactive synthesis algorithms can compute a
\emph{correct-by-construction} implementation of a formal specification fully
automatically.  Besides the construction of full systems~\cite{BloemGJPPW07b},
synthesis algorithms are also used in automatic debugging to compute corrections
of erroneous parts of a design~\cite{StaberB07}, or in program sketching, where
``holes'' (parts that are left blank by the designer) are filled
automatically~\cite{Solar-Lezama09}.

This work deals with synthesis of hardware systems from safety specifications.
Safety specifications express that certain ``bad things'' never happen.  This is
an important class of specifications for two reasons.  First, bounded synthesis
approaches~\cite{Ehlers10} can reduce synthesis from richer specifications to
safety synthesis problems.  Second, safety properties often make up the bulk of
a specification, and they can be handled in a compositional manner: the safety
synthesis problem can be solved before the other properties are
handled~\cite{SohailS09}.

One challenge for reactive synthesis is scalability. To address it, synthesis
algorithms are usually symbolic, i.e., they represent states and transitions
using formulas.  The symbolic representations are, in turn, often implemented
using Binary Decision Diagrams (BDDs), because they provide both existential and
universal quantification.  However, it is well known that BDDs explode in size
for certain structures~\cite{BiereCCZ99}.  At the same time, algorithms and
tools to decide the satisfiability of formulas became very efficient over the
last decade.

In this paper, we thus propose several new approaches to use
satisfiability-based methods for the synthesis of reactive systems from safety
specifications.  We focus on the computation of the so-called \emph{winning
region}, i.e., the states from which the specification can be fulfilled, because
extracting an implementation from this winning region is then conceptually easy
(but can be computationally hard).  More specifically, our contributions are as
follows.
\begin{enumerate}
\item
We present a learning-based approach to compute a winning region as a
Conjunctive Normal Form (CNF) formula over the state variables using a solver
for Quantified Boolean Formulas (QBFs)~\cite{LonsingB10}.
\item 
We show how this method can be implemented efficiently using two incremental
SAT-solvers instead of a QBF-solver, and how approximate reachability
information can be used to increase the performance.  We also present a
parallelization that combines different variants of these learning-based
approaches to achieve a super-linear speedup.
\item 
We present a template-based approach to compute a winning region that follows a
given structure with one single QBF-solver call.
\item
We also show that fixing a structure can be avoided when using a solver for
Effectively Propositional Logic (EPR)~\cite{Lewis80}.
\item 
We present extensive experimental results to compare all these methods, to each
other and to previous work.
\end{enumerate}
Our experiments do not reveal \emph{the} new all-purpose synthesis algorithm. We
rather conclude that different methods perform well on different benchmarks, and
that our new approaches outperform existing ones significantly on some classes
of benchmarks.

\textbf{Related Work.} A QBF-based synthesis method for safety specifications
was presented in~\cite{StaberB07}.  Its QBF-encoding can have deep quantifier
nestings and many copies of the transition relation.  In contrast, our approach
uses more but potentially cheaper QBF-queries.
Becker et al.~\cite{BeckerELM12} show how to compute all solutions to a
QBF-problem with computational learning, and how to use such an ALLQBF engine
for synthesis.  In order to compute all losing states (from which the
specification cannot be enforced) their algorithm analyzes all one-step
predecessors of the unsafe states before turning to the two-step predecessors,
an so on. Our learning-based synthesis method is similar, but applies learning
directly to the synthesis problem.  As a result, our synthesis algorithm is more
``greedy''.  Discovered losing states are utilized immediately in the
computation of new losing states, independent of the distance to the unsafe
states. Besides the computation of a winning region, computational learning has
also been used for extracting small circuits from a strategy~\cite{EhlersKH12}.
The basic idea of substituting a QBF-solver with two competing SAT-solvers has
already been presented in~\cite{JanotaS11} and~\cite{MorgensternGS13}. We apply
this idea to our learning-based synthesis algorithm, and adapt it to make
optimal use of incremental SAT-solving in our setting.  Our optimizations to
utilize reachability information in synthesis are based on the concept of
incremental induction, as presented by Bradley for the model-checking algorithm
IC3~\cite{Bradley11}.  These reachability optimizations are completely new in
synthesis, to the best of our knowledge.
Recently, Morgenstern et al.~\cite{MorgensternGS13} proposed a property-directed
synthesis method which is also inspired by IC3~\cite{Bradley11}. Roughly
speaking, it computes the rank (the number of steps in which the environment can
enforce to reach an unsafe state) of the initial state in a lazy manner.  It
maintains over-approximations of states having (no more than) a certain rank. If
the algorithm cannot decide the rank of a state using this information, it
decides the rank of successors first.  This approach is complementary to our
learning-based algorithms.  One fundamental difference is
that~\cite{MorgensternGS13} explores the state space starting from the initial
state, while our algorithms start at the unsafe states.  The main similarity is
that one of our methods also uses two competing SAT-solvers instead of a
QBF-solver.
Templates have already been used to synthesize combinational
circuits~\cite{KojevnikovKY09}, loop invariants~\cite{ErnstPGMPTX07},
repairs~\cite{KonighoferB11}, and missing parts in
programs~\cite{Solar-Lezama09}.  We use this idea for synthesizing a winning
region.  Reducing the safety synthesis problem to EPR is also new, to the best
of our knowledge.

\textbf{Outline.} The rest of this paper is organized as follows.
Section~\ref{sec:prel} introduces basic concepts and notation, and
Section~\ref{sec:synth} discusses synthesis from safety specifications in
general.  Our new synthesis methods are presented in
Sections~\ref{sec:learn_synth} and~\ref{sec:direct}. Section~\ref{sec:exp_res}
contains our experimental evaluation, and Section~\ref{sec:concl} concludes.
This is an extended version of~\cite{BloemKS14}, featuring an additional
appendix.

\section{Preliminaries}
\label{sec:prel}

We assume familiarity with propositional logic, but repeat the notions important
for this paper.  Refer to~\cite{hos} for a more gentle introduction.

\textbf{Basic Notation.}
In propositional logic, a \emph{literal} is a Boolean variable or its negation.
A \emph{cube} is a conjunction of literals, and a \emph{clause} is a disjunction
of literals.  A formula in propositional logic is in \emph{Conjunctive Normal
Form (CNF)} if it is a conjunction of clauses. A cube describes a (potentially
partial) assignment to Boolean variables: unnegated variables are $\true$,
negated ones are $\false$. We denote vectors of variables with overlines, and
corresponding cubes in bold. E.g., $\mathbf{x}$ is a cube over the variable
vector $\overline{x}=(x_1,\ldots,x_n)$. We treat vectors of variables like sets
if the order does not matter. An $\overline{x}$-\emph{minterm} is a cube that
contains all variables of $\overline{x}$. Cube $\mathbf{x}_1$ is a
\emph{sub-cube} of $\mathbf{x}_2$, written $\mathbf{x}_1 \subseteq
\mathbf{x}_2$, if the literals of $\mathbf{x}_1$ form a subset of the literals
in $\mathbf{x}_2$.  We use the same notation for \emph{sub-clauses}. Let
$F(\overline{x})$ be a propositional formula over the variables $\overline{x}$,
and let $\mathbf{x}$ be an $\overline{x}$-minterm.  We write $\mathbf{x} \models
F(\overline{x})$ to denote that the assignment $\mathbf{x}$ satisfies
$F(\overline{x})$.  We will omit the brackets listing variable dependencies if
they are irrelevant or clear from the context (i.e., we often write $F$ instead
of $F(\overline{x})$).

\textbf{Decision Procedures.}
A \emph{SAT-solver} is a tool that takes a propositional formula (usually in
CNF) and decides its satisfiability. Let $F(\overline{x}, \overline{y}, \ldots)$
be a propositional formula over several vectors $\overline{x}, \overline{y},
\ldots$ of Boolean variables.  We write $\textsf{sat} := \propsat(F)$ for a
SAT-solver call. The variable $\textsf{sat}$ is assigned $\true$ if $F$ is
satisfiable, and $\false$ otherwise.  We write $(\textsf{sat}, \mathbf{x},
\mathbf{y}, \ldots) := \propsatmodel(F(\overline{x}, \overline{y}, \ldots))$ to
obtain a satisfying assignment in the form of cubes $\mathbf{x}, \mathbf{y},
\ldots$ over the different variable vectors.  Let $\mathbf{a}$ be a cube. We
write $\mathbf{b} := \propsatcore(\mathbf{a}, F)$ to denote the extraction of an
unsatisfiable core: Given that $\mathbf{a} \wedge F$ is unsatisfiable,
$\mathbf{b} \subseteq \mathbf{a}$ will be a sub-cube of $\mathbf{a}$ such that
$\mathbf{b} \wedge F$ is still unsatisfiable.
\emph{Quantified Boolean Formulas (QBFs)} extend propositional logic with
universal ($\forall$) and existential ($\exists$) quantifiers.  A QBF (in Prenex
Conjunctive Normal Form) is a formula $Q_1\overline{x} \scope Q_2\overline{y}
\scope \ldots F(\overline{x},\overline{y}, \ldots)$, where $Q_i \in \{\forall,
\exists\}$ and $F$ is a propositional formula in CNF.  Here, $Q_i\overline{x}$
is a shorthand for $Q_i x_1\ldots Q_i x_n$ with $\overline{x} = (x_1 \ldots
x_n)$.  The quantifiers have their expected semantics. A \emph{QBF-solver} takes
a QBF and decides its satisfiability. We write
$\textsf{sat} := \qbfsat(Q_1\overline{x}
\scope Q_2\overline{y} \scope \ldots F(\overline{x},\overline{y}, \ldots))$ or
$(\textsf{sat}, \mathbf{a}, \mathbf{b} \ldots) :=
\qbfsatmodel(
\exists\overline{a} \scope
\exists\overline{b} \ldots
Q_1 \overline{x} \scope
Q_2 \overline{y} \ldots
F(\overline{a},\overline{b}, \ldots,
\overline{x},\overline{y}, \ldots))$
to denote calls to a QBF-solver.  Note that $\qbfsatmodel$ only extracts
assignments for variables that are quantified existentially on the outermost
level.

\textbf{Transition Systems.}
A \emph{controllable finite-state transition system} is a tuple $\mathcal{S} =
(\overline{x}, \overline{i}, \overline{c}, I, T)$, where $\overline{x}$ is a
vector of Boolean state variables, $\overline{i}$ is a vector of uncontrollable
input variables, $\overline{c}$ is a vector of controllable input variables,
$I(\overline{x})$ is an initial condition, and
$T(\overline{x},\overline{i},\overline{c},\overline{x}')$ is a transition
relation with $\overline{x}'$ denoting the next-state copy of $\overline{x}$.  A
\emph{state} of $\mathcal{S}$ is an assignment to the $\overline{x}$-variables,
usually represented as $\overline{x}$-minterm $\mathbf{x}$.  A formula
$F(\overline{x})$ represents the set of all states $\mathbf{x}$ for which
$\mathbf{x}\models F(\overline{x})$.  Priming a formula $F$ to obtain $F'$ means
that all variables in the formula are primed, i.e., replaced by their next-state
copy. An \emph{execution} of $\mathcal{S}$ is an infinite sequence
$\mathbf{x}_0, \mathbf{x}_1 \ldots $ of states such that $\mathbf{x}_0 \models
I$ and for all pairs $(\mathbf{x}_j, \mathbf{x}_{j+1})$ there exist some input
assignment $\mathbf{i}_j,\mathbf{c}_j$ such that $\mathbf{x}_j \wedge
\mathbf{i}_j \wedge \mathbf{c}_j \wedge \mathbf{x}_{j+1}' \models T$.  A state
$\mathbf{x}$ is \emph{reachable} in $\mathcal{S}$ if there exists an execution
$\mathbf{x}_0, \mathbf{x}_1 \ldots $ and an index $j$ such that $\mathbf{x} =
\mathbf{x}_j$. The execution of $\mathcal{S}$ is controlled by two
\emph{players}: the \emph{protagonist} and the \emph{antagonist}.  In every step
$j$, the antagonist first chooses an assignment $\mathbf{i}_j$ to the
uncontrollable inputs $\overline{i}$.  Next, the protagonist picks an assignment
$\mathbf{c}_j$ to the controllable inputs $\overline{c}$.  The transition
relation $T$ then computes the next state $\mathbf{x}_{j+1}$.  This is repeated
indefinitely. We assume that $T$ is \emph{complete} and \emph{deterministic},
i.e., for every state and input assignment, there exists exactly one successor
state.  More formally, we have that
$\forall \overline{x}, \overline{i}, \overline{c} \scope \exists
\overline{x'} \scope T$
and
$\forall \overline{x}, \overline{i}, \overline{c},
\overline{x_1}', \overline{x_2}' \scope
(T(\overline{x},\overline{i},\overline{c},\overline{x_1}') \wedge
T(\overline{x},\overline{i},\overline{c},\overline{x_2}'))
\Rightarrow (\overline{x_1}' = \overline{x_2}').
$
Let $F(\overline{x})$ be a formula representing a certain set of states. The
mixed pre-image $\FS{1}(F) = \forall \overline{i} \scope \exists
\overline{c},\overline{x}' \scope T \wedge F'$ represents all states from which
the protagonist can enforce to reach a state of $F$ in exactly one step.
Analogously, $\FE{1}(F) = \exists \overline{i} \scope \forall \overline{c}\scope
\exists \overline{x}' \scope T \wedge F'$ gives all states from which the
antagonist can enforce to visit $F$ in one step.

\textbf{Synthesis Problem.}
A (memoryless) \emph{controller} for $\mathcal{S}$ is a function $f:
2^{\overline{x}} \times 2^{\overline{i}} \rightarrow 2^{\overline{c}}$ to define
the control signals $\overline{c}$ based on the current state of $\mathcal{S}$
and the uncontrollable inputs $\overline{i}$. Let $P(\overline{x})$ be a formula
characterizing the set of safe states in a transition system $\mathcal{S}$. An
execution $\mathbf{x}_0, \mathbf{x}_1 \ldots $ is \emph{safe} if it visits only
safe states, i.e., $\mathbf{x}_j \models P$ for all $j$. A controller $f$ for
$\mathcal{S}$ is \emph{safe} if all executions of $\mathcal{S}$ are safe, given
that the control signals are computed by $f$.  Formally, $f$ is safe if there
exists no sequence of pairs $(\mathbf{x}_0, \mathbf{i}_0), (\mathbf{x}_1,
\mathbf{i}_1), \ldots$ such that (a) $\mathbf{x}_0 \models I$, (b) $\mathbf{x}_j
\wedge \mathbf{i}_j \wedge f(\mathbf{x}_j, \mathbf{i}_j) \wedge
\mathbf{x}_{j+1}' \models T$ for all $j \ge 0$, and (c) $\mathbf{x}_j\not\models
P$ for some $j$. The problem addressed in this paper is to synthesize such a
safe controller.  We call a pair $(\mathcal{S},P)$ a \emph{specification} of a
safety synthesis problem.  A specification is \emph{realizable} if a safe
controller exists. A \emph{safe implementation} $\mathcal{I}$ of a specification
$(\mathcal{S},P)$ with $\mathcal{S} = (\overline{x}, \overline{i}, \overline{c},
I(\overline{x}), T(\overline{x}, \overline{i}, \overline{c}, \overline{x}'))$ is
a transition system $\mathcal{I} = (\overline{x}, \overline{i}, \emptyset,
I(\overline{x}), T(\overline{x}, \overline{i}, f(\overline{x}, \overline{i}),
\overline{x}'))$, where $f$ is a safe controller for $\mathcal{S}$.

\section{Synthesis from Safety Specifications}
\label{sec:synth}

This paper presents several approaches for synthesizing a safe controller for a
fine-state transition system $\mathcal{S}$. The synthesis problem can be seen as
a game between the protagonist controlling the $\overline{c}$-variables and the
antagonist controlling the $\overline{i}$-variables during an
execution~\cite{MorgensternGS13}. The protagonist wins the game if the execution
never visits an unsafe state $\mathbf{x}\not\models P$.  Otherwise, the
antagonist wins. A safe controller for $\mathcal{S}$ is now simply a strategy
for the protagonist to win the game. Standard game-based synthesis methods can
be used to compute such a winning strategy~\cite{Thomas95}.  These game-based
methods usually work in two steps. First, a so-called \emph{winning region} is
computed.  A winning region is a set of states $W(\overline{x})$ from which a
winning strategy for the protagonist exists.  Second, a winning strategy is
derived from (intermediate results in the computation of) the winning region.
Most of the synthesis approaches presented in the following implement this
two-step procedure. For safety synthesis problems, the following three
conditions are sufficient for a winning region $W(\overline{x})$ to be turned
into a winning strategy.
\begin{enumerate}[I)]
\item Every initial state is in the winning region: $I \Rightarrow W$.
\item The winning region contains only safe states: $W \Rightarrow P$.
\item The protagonist can enforce to stay in the winning region:
$W \Rightarrow \FS{1}(W)$.
\end{enumerate}
A specification is realizable if and only if such a winning region exists.
Hence, it suffices to search for a formula that satisfies these three
constraints.
Deriving a winning strategy $f: 2^{\overline{x}} \times 2^{\overline{i}}
\rightarrow 2^{\overline{c}}$ from such a winning region is then conceptually
easy: $f$ must always pick control signal values such that the successor state
is in $W$ again.  This is always possible due to (I) and (III). We therefore
focus on approaches to efficiently compute a winning region that satisfies
(I)-(III), and leave an investigation of methods for the extraction of a
concrete controller to future work\footnote{In our implementation, we currently
extract circuits by computing Skolem functions for the $\overline{c}$ signals in
$
\forall \overline{x},\overline{i} \scope
\exists \overline{c},\overline{x}'\scope
(\neg W) \vee (T \wedge W')
$
using the \qbfcert~\cite{NiemetzPLSB12} framework.  However, there are other
options like learning~\cite{EhlersKH12}, interpolation~\cite{JiangLH09}, or
templates~\cite{KojevnikovKY09}.}. First, we will briefly discuss an
attractor-based approach which is often implemented with BDDs~\cite{Thomas95}.
Then, we will present several new ideas which are more suitable for an
implementation using SAT- and QBF-solvers.

\subsection{Standard Attractor-Based Synthesis Approach}
\label{sec:std}
The synthesis method presented in this section can be seen as the standard
textbook method for solving safety games~\cite{Thomas95}.  Starting with all
safe states $P$,
\begin{wrapfigure}[9]{r}{0.45\textwidth}
\vspace{-0.7cm}
\begin{algorithmic}[1]
\ProcedureLi{SafeSynth}
             {$\mathcal{S},P$}
             {$W$ or $\false$}
  \State $F := P$
  \While{$F$ changes}
    \State $F := F \wedge \FS{1}(F)$
    \If{$I \not \Rightarrow F$}
      \State \textbf{return} $\false$
    \EndIf
  \EndWhile
  \State \textbf{return} $F$
\EndProcedure  
\end{algorithmic}
\end{wrapfigure}
the \textsc{SafeSynth} algorithm reduces $F$ to states from which the
protagonist can enforce to go back to $F$ until $F$ does not change anymore.  If
an initial state is removed from $F$, $\false$ is returned to signal
unrealizability.  Otherwise, $F$ will finally converge to a fixpoint, which is a
proper winning region $W$ ($W = \nu F. P \wedge \FS{1}(F)$ in $\mu$-calculus
notation). \textsc{SafeSynth} is well suited for an implementation using BDDs
because the set of all states satisfying $\FS{1}(F)$ can be computed with just a
few BDD operations, and the comparison to decide if $F$ changed can be done in
constant time.
A straightforward implementation using a QBF-solver maintains a growing
quantified formula to represent $F$ (i.e,
$F_0 = P$,
$F_1 = \exists \overline{x} \scope
       \forall \overline{i} \scope
       \exists \overline{c}, \overline{x}'\scope
       P \wedge T \wedge P'$,
and so on), and calls a QBF-solver to decide if $F$ changed semantically from
one iteration to the next one.  This approach is explained in~\cite{StaberB07}.
In iteration $n$, $F$ contains $n$ copies of the transition relation and $2n$
quantifier alternations.  This means that the difficulty of the QBF queries
increases significantly with the number of iterations, which may be prohibitive
for large specification.  The resulting winning region $W$ is a quantified
formula as well.
An alternative QBF-based implementation~\cite{BeckerELM12} eliminates the
quantifiers from $F$ in every iteration by computing all satisfying assignments
of $F$.  The next section explains how this idea can be improved.

\section{Learning-Based Synthesis Approaches}
\label{sec:learn_synth}

\begin{figure}[tb]
 \begin{minipage}{.3\textwidth}
  \begin{center}
    \includegraphics[width=\textwidth]{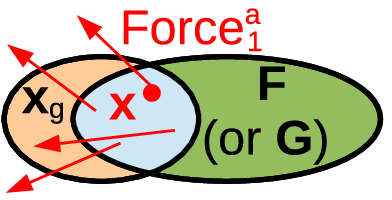}
    \caption{\textsc{LearnQbf}: working principle.}
    \label{fig:learn_qbf}
  \end{center}
 \end{minipage}
 \hspace{0.4cm}
 \begin{minipage}{.3\textwidth}
   \begin{center}
    \includegraphics[width=\textwidth]{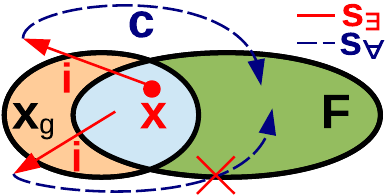}
    \caption{\textsc{LearnSat}: working principle.}
    \label{fig:learn_sat}
  \end{center}
 \end{minipage}
 \hspace{0.4cm}
 \begin{minipage}{.3\textwidth}
   \begin{center}
    \includegraphics[width=\textwidth]{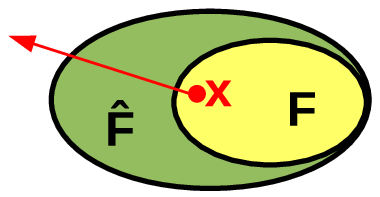}
    \caption{\textsc{LearnSat}: Using $\hat{F}$ for incremental solving.}
    \label{fig:learn_sat2}
  \end{center}
 \end{minipage}
\end{figure}

Becker et al.~\cite{BeckerELM12} show how \textsc{SafeSynth} can be implemented
with a QBF-solver by eliminating the quantifiers in $F$ with computational
learning.  This gives a CNF representation of every $F$-iterate.  However, we
are only interested in the final value $W$ of $F$.  This allows for a tighter
and more efficient integration of the learning approach with
the~\textsc{SafeSynth} algorithm.
  
\subsection{Learning-Based Synthesis using a QBF-Solver}
\label{sec:qbf_learn}

The following algorithm uses computational learning to compute a winning region
in CNF using a QBF-solver.  It returns $\false$ in case of unrealizability.

\begin{algorithmic}[1]
\ProcedureRet{LearnQbf}
             {$(\overline{x}, \overline{i}, \overline{c}, I, T),P$}
             {$W$ or $\false$}
  \State $F := P$
  \State // Check if there exists an
         $\mathbf{x} \models F \wedge \FE{1}(\neg F)$:
  \While{$\mathsf{sat}$ with $(\mathsf{sat},\mathbf{x})$:=$\qbfsatmodel(
      \exists \overline{x},\overline{i} \scope
      \forall \overline{c} \scope
      \exists \overline{x}' \scope
      F \wedge T \wedge \neg F')$}\label{alg:learn_ce}
    \State // Find a sub-cube $\mathbf{x}_g \subseteq \mathbf{x}$ such that
           $(\mathbf{x}_g \wedge F) \Rightarrow \FE{1}(\neg F)$:
    \State $\mathbf{x}_g := \mathbf{x}$
    \For{$l\in \textsc{literals}(\mathbf{x})$}
      \State $\mathbf{x}_t := \mathbf{x}_g \setminus \{l\}$,             
             \algorithmicif\ \textsf{optimize} \algorithmicthen\
             $G:= F \wedge \neg \mathbf{x}_g$ \algorithmicelse\
             $G:= F$~\label{alg:learn_op}
      \If{$\neg\qbfsat(
      \exists \overline{x} \scope
      \forall \overline{i} \scope
      \exists \overline{c},\overline{x}' \scope
      \mathbf{x}_t \wedge G \wedge T \wedge G')$}~\label{alg:learn_gen}
        \State $\mathbf{x}_g := \mathbf{x}_t$
      \EndIf
    \EndFor
    \LineIf{$\propsat(\mathbf{x}_g \wedge I)$}{\textbf{return} $\false$}
    \State $F := F \wedge \neg \mathbf{x}_g$\label{alg:learn_rem}
  \EndWhile
  \State \textbf{return } $F$
\EndProcedureRet
\end{algorithmic}
The working principle of \textsc{LearnQbf} is illustrated in
Figure~\ref{fig:learn_qbf}.  It starts with the initial guess $F$ that the
winning region contains all safe states $P$.  Line~\ref{alg:learn_ce} then
checks for a counterexample to the correctness of this guess in form of a state
$\mathbf{x}\models F\wedge \FE{1}(\neg F)$ from which the antagonist can enforce
to leave $F$. Assume that $\textsf{optimize}=\false$ in line~\ref{alg:learn_op}
for now, i.e., $G$ is always just $F$. The inner loop now generalizes the
state-cube $\mathbf{x}$ to $\mathbf{x}_g \subseteq \mathbf{x}$ by dropping
literals as long as $\mathbf{x}_g$ does not contain a single state from which
the protagonist can enforce to stay in $F$. During and after the execution of
the inner loop, $\mathbf{x}_g$ contains only states that must be removed from
$F$, or have already been removed from $F$ before.  Hence, as an optimization,
we can treat the states of $\mathbf{x}_g$ as if they were removed from $F$
already \emph{during} the cube minimization.  This is done with
$\textsf{optimize}=\true$ in line~\ref{alg:learn_op} by setting $G = F \wedge
\neg \mathbf{x}_g$ instead of $G=F$.  This optimization can lead to smaller
cubes and less iterations.
If the final cube $\mathbf{x}_g$ contains an initial state, the algorithm
signals unrealizability by returning $\false$. Otherwise, it removes the states
of $\mathbf{x}_g$ from $F$ by adding the clause $\neg \mathbf{x}_g$, and
continues by checking for other counterexamples.  If $P$ is in CNF, then the
final result in $F$ will also be in CNF.  If $T$ is also in CNF, then the query
of line~\ref{alg:learn_gen} can be constructed by merging clause sets.  Only for
the query in line~\ref{alg:learn_ce}, a CNF encoding of $\neg F'$ is necessary.
This can be achieved, e.g., using a Plaisted-Greenbaum
transformation~\cite{PlaistedG86}, which causes only a linear blow-up of the
formula.

\textbf{Heuristics.} We observed that the generalization (the inner loop of
\textsc{LearnQbf}) is often fast compared to the computation of counterexamples
in Line~\ref{alg:learn_ce}.  As a heuristic, we therefore propose to compute not
only one but all (or several) minimal generalizations $\mathbf{x}_g\subseteq
\mathbf{x}$ to every counterexample-state $\mathbf{x}$, e.g., using a hitting
set tree algorithm~\cite{Reiter87}.  Another observation is that newly
discovered clauses can render earlier clauses redundant in $F$. In every
iteration, we therefore ``compress'' $F$ by removing clauses that are implied by
others.  This can be done cheaply with incremental SAT-solving, and simplifies
the CNF for $\neg F'$ in line~\ref{alg:learn_ce}. Iterating over existing
clauses and trying to minimize them further at a later point in time did not
lead to significant improvements in our experiments.

\subsection{Learning-Based Synthesis using SAT-Solvers}
\label{sec:sat_learn}
\textsc{LearnQbf} can also be implemented with SAT-solving instead of
QBF-solving.  The basic idea is to use two competing SAT-solvers for the two
different quantifier types, as done in~\cite{JanotaS11}.  However, we interweave
this concept with the synthesis algorithm to better utilize incremental solving
capabilities of modern SAT-solvers.
\begin{algorithmic}[1]
\ProcedureRet{LearnSat}
             {$(\overline{x}, \overline{i}, \overline{c}, I, T),P$}
             {$W$ or $\false$}
  \State $F := P$, $\hat{F} := P$, $U := \true$, $\textsf{precise} := \true$
  \While{$\true$}
    \State $(\mathsf{sat},\mathbf{x},\mathbf{i}) := \propsatmodel(
            F \wedge U \wedge T \wedge \neg \hat{F}')$ \label{alg:sa0}
    \If{$\neg \mathsf{sat}$}
      \LineIf{$\textsf{precise}$}{\textbf{return} $F$}
      \State $U := \true$, $\hat{F} := F$, $\textsf{precise} := \true$
      \label{alg:sar}
    \Else
      \State $(\mathsf{sat},\mathbf{c}) := \propsatmodel(
              F \wedge \mathbf{x} \wedge \mathbf{i} \wedge T \wedge F')$
              \label{alg:sa1}
      \If{$\neg \mathsf{sat}$}
        \State $\mathbf{x}_g := \propsatcore(\mathbf{x},
        F \wedge \mathbf{i} \wedge T \wedge F')$ \label{alg:sa2}
        \LineIf{$\propsat(\mathbf{x}_g \wedge I)$}
        {\textbf{return} $\false$} \label{alg:saun}
        \State $F := F \wedge \neg \mathbf{x}_g$ \label{alg:updf}
        \LineIfElse{\textsf{optimize}}{$\textsf{precise}:= \false$}
        {$\hat{F} := F$, $U:= \true$}
        \label{alg:sanc}
      \Else
        \State $U := U \wedge \neg \propsatcore(
        \mathbf{x} \wedge \mathbf{i},
        \mathbf{c} \wedge F \wedge U \wedge T \wedge \neg \hat{F}')$
        \label{alg:sa3}
      \EndIf
    \EndIf
  \EndWhile
\EndProcedureRet
\end{algorithmic}

\textbf{Data Structures.} Besides the current guess $F$ of the winning region
$W$, \textsc{LearnSat} also maintains a copy $\hat{F}$ of $F$ that is updated
only lazily.  This allows for better utilization of incremental SAT-solving, and
will be explained below. The flag $\textsf{precise}$ indicates if $\hat{F}=F$.
The variable $U$ stores a CNF formula over the $\overline{x}$ and $\overline{i}$
variables.  Intuitively, $U$ contains state-input combinations which are not
useful for the antagonist when trying to break out of $F$.

\textbf{Working Principle.}
The working principle of \textsc{LearnSat} is illustrated in
Figure~\ref{fig:learn_sat}.  For the moment, let \textsf{optimize} be $\false$,
i.e., $\hat{F}$ is always $F$.  To deal with the mixed quantification inherent
in synthesis, \textsc{LearnSat} uses two competing SAT-solvers, $s_\exists$ and
$s_\forall$. In line~\ref{alg:sa0}, $s_\exists$ tries to find a possibility for
the antagonist to leave $F$.  It is computed as a state-input pair
$(\mathbf{x},\mathbf{i})$ for which some $\overline{c}$-value leads to a $\neg
F$ successor.  Next, in line~\ref{alg:sa1}, $s_\forall$ searches for a response
$\mathbf{c}$ of the protagonist to avoid leaving $F$. If no such response
exists, then $\mathbf{x}$ must be excluded from $F$. However, instead of
excluding this one state only, we generalize the state-cube $\mathbf{x}$ by
dropping literals to obtain $\mathbf{x}_g$, representing a larger region of
states for which input $\mathbf{i}$ can be used by the antagonist to enforce
leaving $F$.  This is done by computing the unsatisfiable core with respect to
the literals of $\mathbf{x}$ in line~\ref{alg:sa2}.  Otherwise, if $s_\forall$
finds a response $\mathbf{c}$, then the state-input pair
$(\mathbf{x},\mathbf{i})$ is not helpful for the antagonist to break out of $F$.
 It must be removed from $U$ to avoid that the same pair is tried again.
Instead of removing just $(\mathbf{x},\mathbf{i})$, we generalize it again by
dropping literals as long as the control value $\mathbf{c}$ prevents leaving
$F$.  This is done by computing an unsatisfiable core over the literals in
$\mathbf{x} \wedge \mathbf{i}$ in line~\ref{alg:sa3}.\\ As soon as $F$ changes,
$U$ must be reset to $\true$ (line~\ref{alg:sanc}): even if a state-input pair
is not helpful for breaking out of $F$, it may be helpful for breaking out of a
smaller $F$. If line~\ref{alg:sa0} reports unsatisfiability, then the antagonist
cannot enforce to leave $F$, i.e., $F$ is a winning region ($\textsf{precise} =
\true$ if $\textsf{optimize} = \false$).  If an initial state is removed from
$F$, then the specification is unrealizable (line~\ref{alg:saun}).

\textbf{Using $\hat{F}$ to Support Incremental Solving.}
Now consider the case where \textsf{optimize} is $\true$.  In
line~\ref{alg:updf}, new clauses are added only to $F$ but not to $\hat{F}$.
This ensures that $F \Rightarrow \hat{F}$, but $F$ can be strictly stronger.
See Figure~\ref{fig:learn_sat2} for an illustration.  Line~\ref{alg:sa0} now
searches for a transition (respecting $U$) from $F$ to $\neg \hat{F}$.  If such
a transition is found, then it also leads from $F$ to $\neg F$.  However, if no
such transition from $F$ to $\neg \hat{F}$ exists, then this does not mean that
there is no transition from $F$ to $\neg F$.  Hence, in case of
unsatisfiability, we update $\hat{F}$ to $F$ and store the fact that $\hat{F}$
is now accurate by setting $\textsf{precise} = \true$.  If the call in
line~\ref{alg:sa0} reports unsatisfiability with $\textsf{precise} = \true$,
then there is definitely no way for the antagonist to leave $F$ and the
computation of $F$ is done.  The reason for not updating $\hat{F}$ immediately
is that solver $s_\exists$ can be used incrementally until the next update,
because new clauses are only added to $F$ and $U$.  Only when reaching
line~\ref{alg:sar}, a new incremental session has to be started.  This
optimization proved to be very beneficial in our experiments.  Solver
$s_\forall$ can be used incrementally throughout the entire algorithm anyway,
because $F$ gets updated with new clauses only.

\subsection{Utilizing Unreachable States}
\label{sec:reach}

This section presents an optimization of \textsc{LearnQbf} to utilize
(un)reachability information.  It works analogously for \textsc{LearnSat},
though.  Recall that the variable $G$ in \textsc{LearnQbf} stores the current
over-approximation of the winning region $W$ (cf.~Section.~\ref{sec:qbf_learn}).
\textsc{LearnQbf} generalizes a counterexample-state $\mathbf{x}$ to a region
$\mathbf{x}_g$ such that $G\wedge \mathbf{x}_g \Rightarrow \FE{1}(\neg G)$,
i.e., $G\wedge \mathbf{x}_g$ contains only states from which the antagonist can
enforce to leave $G$.  Let $R(\overline{x})$ be an over-approximation of the
states reachable in $\mathcal{S}$.  That is, $R$ contains at least all states
that could appear in an execution of $\mathcal{S}$.  It is sufficient to ensure
$G\wedge \mathbf{x}_g \wedge R \Rightarrow \FE{1}(\neg G)$ because unreachable
states can be excluded from $G$ even if they are winning for the protagonist.
This can lead to smaller cubes and faster convergence.

There exist various methods to compute reachable states, both precisely and as
over-approximation~\cite{MoonKSS99}.  The current over-approximation $G$ of the
winning region $W$ can also be used: Given that the specification is realizable
(we will discuss the unrealizable case below), the protagonist will enforce that
$W$ is never left.   Hence, at any point in time, $G$ is itself an
over-approximation of the reachable states, not necessarily in $\mathcal{S}$,
but definitely in the final implementation $\mathcal{I}$ (given that
$\mathcal{I}$ is derived from $W$ and $W\Rightarrow G$).  Hence, stronger
reachability information can be obtained by considering only transitions that
remain in $G$.

In our optimization, we do not explicitly compute an over-approximation of the
reachable states, but rather exploit ideas from the property directed
reachability algorithm IC3~\cite{Bradley11}:  By induction, we know that a state
$\mathbf{x}$ is definitely unreachable in $\mathcal{I}$ if
$\mathbf{x}\not\models I$ and $\neg \mathbf{x} \wedge G \wedge T \Rightarrow
\neg \mathbf{x}'$. Otherwise, $\mathbf{x}$ could be reachable. The same holds
for sets of states. By adding these two constraints, we modify the
generalization check in line~\ref{alg:learn_gen} of \textsc{LearnQbf} to
\begin{align}
\qbfsat(
&
      \exists \overline{x}^*, \overline{i}^*, \overline{c}^*\scope
      \exists \overline{x} \scope
      \forall \overline{i} \scope
      \exists \overline{c},\overline{x}' \scope \notag
\\&
      (I(\overline{x}) \vee
       G(\overline{x}^*) \wedge
       \neg \mathbf{x}_g(\overline{x}^*) \wedge
       T(\overline{x}^*, \overline{i}^*, \overline{c}^*, \overline{x}))
       \wedge \label{eq:gen_reach}
\\&
      \mathbf{x}_g(\overline{x}) \wedge
      G(\overline{x}) \wedge
      T(\overline{x},\overline{i},\overline{c},\overline{x}') \wedge
      G(\overline{x}')).\notag
\end{align}
We will refer to this modification as optimization $\textsf{RG}$ (which is short
for ``reachability during generalization''). Only the second line is new.  Here,
$\overline{x}^*$, $\overline{i}^*$, and $\overline{c}^*$ are the previous-state
copies of $\overline{x}$, $\overline{i}$, and $\overline{c}$, respectively.
Originally, the formula was $\true$ if the region $\mathbf{x}_g \wedge G$
contained a state from which the protagonist could enforce to stay in $G$. In
this case, the generalization failed, because we cannot safely remove states
that are potentially winning for the protagonist. The new formula is $\true$
only if $\mathbf{x}_g \wedge G$ contains a state $\mathbf{x}_a$ from which the
protagonist can enforce to stay in $G$, and this state $\mathbf{x}_a$ is either
initial, or has a predecessor $\mathbf{x}_b$ in $G \wedge \neg \mathbf{x}_g$.
This situation is illustrated in Figure~\ref{fig:gen_reach}. States that are
neither initial nor have a predecessor in $G \wedge \neg \mathbf{x}_g$ are
unreachable and, hence, can safely be removed. Note that we require
$\mathbf{x}_b$ to be in $G \wedge \neg \mathbf{x}_g$, and not just in $G$ and
different from $\mathbf{x}_a$. The intuitive reason is that a predecessor in $G
\wedge \mathbf{x}_g$ does not count because this region is going to be removed
from $G$.  A more formal argument is given by the following theorem.
\begin{theorem}
For a realizable specification, if Eq.~\ref{eq:gen_reach} is unsatisfiable, then
$G \wedge \mathbf{x}_g$ cannot contain a state $\mathbf{x}_a$ from which (a) the
protagonist can enforce to visit $G$ in one step, and (b) which is reachable in
any implementation $\mathcal{I}$ derived from a winning region $W\Rightarrow G$
with $W\Rightarrow \FS{1}(W)$. \label{th:gen_reach1}
\end{theorem}
A proof can be found in Appendix~\ref{sec:appendix_proof}.
Theorem~\ref{th:gen_reach1} ensures that the states removed with optimization
\textsf{RG} cannot be necessary for the protagonist to win the game, i.e., that
the optimization does not remove ``too much''.  So far, we assumed
realizability. However, optimization \textsf{RG} also cannot make an
unrealizable specification be identified as realizable.  It can only remove more
states, which means that unrealizability is detected only earlier.

Similar to improving the generalization of counterexamples using unreachability
information, we can also restrict their computation to potentially reachable
states.  This is explained as optimization \textsf{RC} in
Appendix~\ref{sec:appendix_rc}.  However, while optimization \textsf{RG}
resulted in significant performance gains (more than an order of magnitude for
some benchmarks; see the columns SM and SGM in Table~\ref{table:results2}), we
could not achieve solid improvements with optimization \textsf{RC}.  Sometimes
the computation became slightly faster, sometimes slower.

\begin{figure}[tb]
\centering
\begin{minipage}{.30\textwidth}
   \centering
   \includegraphics[width=\textwidth]{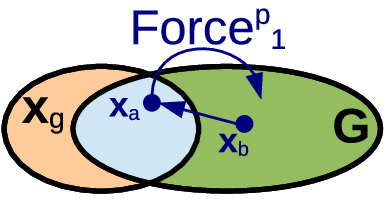}
   \caption{Optimization \textsf{RG}: A counterexample to generalization.}
   \label{fig:gen_reach}
\end{minipage}%
\hspace{0.5cm}
\begin{minipage}{.65\textwidth}
 \centering
 \includegraphics[width=\textwidth]{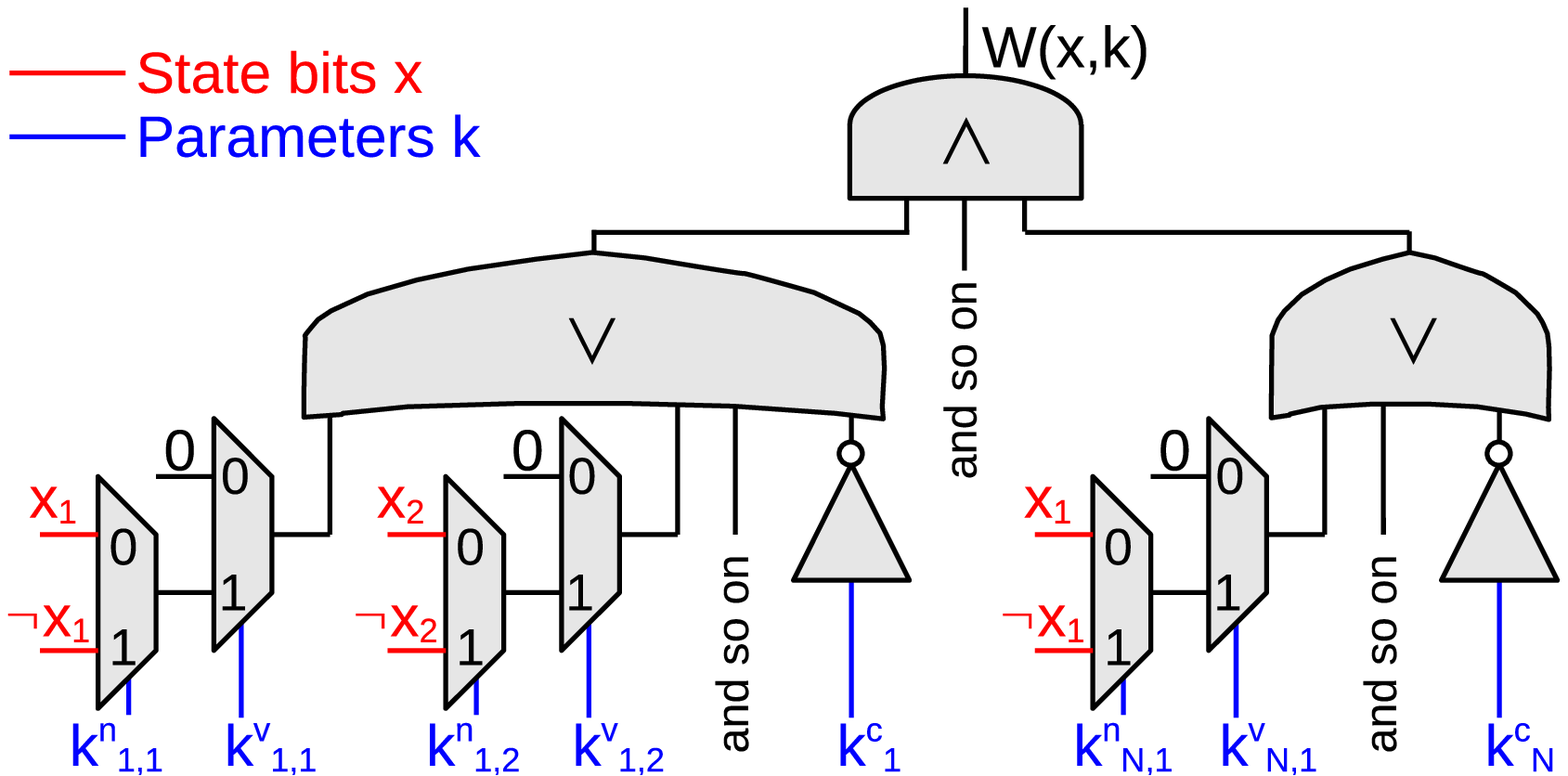}
 \caption{A CNF template for the winning region.}
 \label{fig:cnf_template}
\end{minipage}%
\end{figure}

\subsection{Parallelization}
\label{sec:parallel}
The algorithms \textsc{LearnQbf} and \textsc{LearnSat} compute clauses that
refine the current over-approximation $F$ of the winning region.  This can also
be done with multiple threads in parallel using a global clause database $F$.
Different threads can implement different methods to compute new clauses, or
generalize existing ones.  They notify each other whenever they add a (new or
smaller) clause to $F$ so that all other threads can continue to work with the
refined $F$.

In our implementation, we experimented with different thread combinations. If
two threads are available, we let them both execute \textsc{LearnSat} with
optimization \textsf{RG} but without \textsf{RC}.  We keep the
\textsc{LearnSat}-threads synchronized in the sense that they all use the same
$\hat{F}$.  If one thread restarts solver $s_\exists$ with a new $\hat{F}$, then
all other \textsc{LearnSat}-threads restart their $s_\exists$-solver with the
same $\hat{F}$ as well.  This way, the \textsc{LearnSat}-threads can not only
exchange new $F$-clauses, but also new $U$-clauses.  We use different
SAT-solvers in the different threads (currently our implementation supports
\lingeling, \minisat, and \picosat).  This reduces the chances that the threads
find the same (or similar) counterexamples and generalizations.  Also, the
solvers may complement each other: if one gets stuck for a while on a hard
problem, the other one may still achieve significant progress in the meantime.
The stuck solver then benefits from this progress in the next step.  We also let
the \textsc{LearnSat}-threads store the computed counterexample-cubes in a
global counterexample-database.  If three threads are available, we use one
thread to take counterexample-cubes from this database, and compute all possible
generalizations using a SAT-solver and a hitting set tree
algorithm~\cite{Reiter87}.  We also experimentally added threads that minimize
existing clauses further using a QBF-solver, and threads implementing
\textsc{LearnQbf}.  However, we observed that threads using QBF-solvers can not
quite keep up with the pace of threads using SAT-solvers. Consequently, they
only yield minor speedups.

Our parallelization approach does not only exploit hardware parallelism, it is
also a playground for combining different methods and solvers.  We only tried a
few options; a thorough investigation of beneficial combinations remains to be
done.

\section{Direct Synthesis Methods}
\label{sec:direct}

This section presents completely different approaches for computing a winning
region.  Instead of refining an initial guess in many iterations, we simply
assert the constraints for a proper winning region and compute a solution in one
go.

\subsection{Template-Based Synthesis Approach}

We define a generic template $W(\overline{x},\overline{k})$ for the winning
region $W(\overline{x})$, where $\overline{k}$ is a vector of Boolean variables
acting as template parameters. Concrete values $\mathbf{k}$ for the parameters
$\overline{k}$ instantiate a concrete formula $W(\overline{x})$ over the state
variables $\overline{x}$. This reduces the search for a Boolean formula (the
winning region) to a search for Boolean parameter values. We can now find a
winning region that satisfies the three desired properties (I)-(III) with a
single QBF-solver call:
\begin{eqnarray}
(sat, \mathbf{k}) = \qbfsatmodel(
\exists \overline{k} \scope
\forall \overline{x}, \overline{i} \scope
\exists \overline{c}, \overline{x}' \scope
&& (I \Rightarrow W(\overline{x}, \overline{k})) \wedge \notag\\
&& (W(\overline{x}, \overline{k}) \Rightarrow P) \wedge \label{eq:templ}\\
&& (W(\overline{x}, \overline{k}) \Rightarrow \notag
(T\wedge W(\overline{x}', \overline{k})))\quad \quad
\end{eqnarray}
The challenge in this approach is to define a generic template
$W(\overline{x},\overline{k})$ for the winning region.
Figure~\ref{fig:cnf_template} illustrates how a CNF template could look like.
Here, $W(\overline{x})$ is a conjunction of clauses over the state variables
$\overline{x}$.  Template parameters $\overline{k}$ define the shape of the
clauses.  First, we fix a maximum number $N$ of clauses in the CNF. Then, we
introduce three vectors of template parameters: $\overline{k^c}$,
$\overline{k^v}$, and $\overline{k^n}$.  We denote their union by
$\overline{k}$. If parameter $k^c_i$ with $1\leq i\leq N$ is $\true$, then
clause $i$ is used in $W(\overline{x})$, otherwise not.  If parameter
$k^v_{i,j}$ with $1\leq i\leq N$ and $1\leq j\leq |\overline{x}|$ is $\true$,
then the state variable $x_j\in\overline{x}$ appears in clause $i$ of
$W(\overline{x})$, otherwise not. Finally, if parameter $k^n_{i,j}$ is $\true$,
then $x_j$ can appear in clause $i$ only negated, otherwise only unnegated. If
$k^v_{i,j}$ is $\false$, then $k^n_{i,j}$ is irrelevant.  This gives
$|\overline{k}|=2\cdot N \cdot |\overline{x}| + N$ template parameters.
Figure~\ref{fig:cnf_template} illustrates this definition of
$W(\overline{x},\overline{k})$ as a circuit.  A CNF encoding of this circuit to
be used in the QBF query shown in Eq.~\ref{eq:templ} is straightforward.
Choosing $N$ is delicate.  If $N$ is too low, we will not find a solution, even
if one exists.  If it is too high, we waste computational resources and may
find an unnecessarily complex winning region. In our implementation, we solve
this dilemma by starting with $N=1$ and doubling it upon failure.  We stop if we
get a negative answer for $N \ge 2^{|\overline{x}|}$ (because any Boolean
formula over $\overline{x}$ can be represented in a CNF with $<
2^{|\overline{x}|}$ clauses). The CNF template explained in this paragraph is
just an example.  Other ideas include And-Inverter Graphs with parameterized
interconnects, or other parameterized circuits~\cite{KojevnikovKY09}.

The template-based approach can be good at finding simple winning regions
quickly.  There may be many different winning regions that satisfy the
conditions (I)-(III).  The algorithms \textsc{SafeSynth}, \textsc{LearnQbf} and
\textsc{LearnSat} will always find the largest of these sets (modulo unreachable
states, if used with optimization \textsf{RG} or \textsf{RC}). The
template-based approach is more flexible.  As an extreme example, suppose that
there is only one initial state, it is safe, and the protagonist can enforce to
stay in this state. Suppose further that the largest winning region is
complicated.  The template-based approach may find $W=I$ quickly, while the
other approaches may take ages to compute the largest winning region.  On the
other hand, the template-based approach can be expected to scale poorly if no
simple winning region exists, or if the synthesis problem is even unrealizable.
The issue of detecting unrealizability can be tackled just like in bounded
synthesis~\cite{FiliotJR09}: in parallel to searching for a winning region for
the protagonist, one can also try to find a winning region for the antagonist (a
set of states from which the antagonist can enforce to leave the safe states in
some number of steps).  If a winning region for the antagonist contains an
initial state, unrealizability is detected.

\subsection{EPR Reduction Approach}
\label{sec:epr}
The EPR approach is based on the observation that a winning region
$W(\overline{x})$ satisfying the three requirements (I)-(III) can also be
computed as a Skolem function, without a need to fix a template. However, the
requirement (III) concerns not only $W$ but also its next-state copy $W'$.
Hence, we need a Skolem function for the winning region and its next-state copy,
and the two functions must be consistent. This cannot be formulated as a QBF
problem with a linear quantifier structure, but only using so-called Henkin
Quantifiers\footnote{A winning region is a Skolem function for the Boolean
variable $w$ in the formula
$
\begin{array}{l}
\forall \overline{x} \scope
\exists w \scope
\forall \overline{i} \scope
\exists \overline{c} \scope\\
\forall \overline{x}' \scope
\exists w' \scope
\end{array}
(I \Rightarrow w) \wedge
(w \Rightarrow P) \wedge
((\overline{x} = \overline{x}') \Rightarrow (w = w')) \wedge
(w\wedge T
\Rightarrow w')
$
.}~\cite{FroehlichKB12}, or in the \emph{Effectively Propositional Logic
(EPR)}~\cite{Lewis80} fragment of first-order logic.  Deciding the
satisfiability of formulas with Henkin Quantifiers is NEXPTIME-complete, and
only a few tools exist to tackle the problem~\cite{FroehlichKB12}.  Hence, we
focus on reductions to EPR.
EPR is a subset of first-order logic that contains formulas of the form $\exists
\overline{A}\scope \forall \overline{B} \scope \varphi$, where $A$ and $B$ are
disjoint vectors of variables ranging over some domain $\mathbb{D}$, and
$\varphi$ is a function-free first-order formula in CNF. The formula $\varphi$
can contain predicates, which are (implicitly) existentially quantified.

Recall that we need to find a formula $W(\overline{x})$ such that
$
\forall \overline{x}, \overline{i} \scope
\exists \overline{c}, \overline{x}' \scope
(I\Rightarrow W) \wedge (W\Rightarrow P) \wedge (W\Rightarrow T\wedge W')
$.
In order to get a corresponding EPR formula, we must (a) encode the Boolean
variables using first-order domain variables, (b) eliminate the existential
quantification inside the universal one, and (c) encode the body of the formula
in CNF.  Just like~\cite{SeidlLB12}, we can address (a) by introducing a new
domain variable $Y$ for every Boolean variable $y$, a unary predicate $p$ to
encode the truth value of variables, constants $\top$ and $\bot$ to encode
$\true$ and $\false$, and the axioms $p(\top)$ and $\neg p(\bot)$.  The
existential quantification of the $\overline{x}'$ variables can be turned into a
universal one by turning the conjunction with $T$ into an implication, i.e.,
re-write
$\forall \overline{x}, \overline{i} \scope
\exists \overline{c}, \overline{x}' \scope
W(\overline{x}) \Rightarrow T(\overline{x}, \overline{i}, \overline{c},
\overline{x}') \wedge W(\overline{x}')$
to
$\forall \overline{x}, \overline{i} \scope
\exists \overline{c} \scope \forall \overline{x}' \scope
W(\overline{x}) \wedge T(\overline{x}, \overline{i}, \overline{c},
\overline{x}') \Rightarrow W(\overline{x}')$.
This works because we assume that $T$ is both deterministic and complete.  We
Skolemize the $\overline{c}$-variables $c_1,\ldots,c_n$ by introducing new
predicates $C_1(\overline{X}, \overline{I}),\ldots, C_n(\overline{X},
\overline{I})$.  For $W$, we also introduce a new predicate $W(\overline{X})$.
This gives
\begin{eqnarray*}
\forall \overline{X}, \overline{I}, \overline{X}' \scope &&
(I(\overline{X}) \Rightarrow W(\overline{X})) \quad \wedge \quad
(W(\overline{X})\Rightarrow P(\overline{X})) \quad \wedge \quad
\\&&
(W(X) \wedge T(\overline{X}, \overline{I}, \overline{C}(\overline{X},
\overline{I}), \overline{X}') \Rightarrow W(X'))
\end{eqnarray*}
The body of this formula has to be encoded in CNF, but many first-order theorem
provers and EPR solvers can do this internally.  If temporary variables are
introduced in the course of a CNF encoding, then they have to be Skolemized with
corresponding predicates.  Instantiation-based EPR-solvers like
\iprover~\cite{Korovin08} can not only decide the satisfiability of EPR
formulas, but also compute models in form of concrete formulas for the
predicates.  For our problem, this means that we cannot only directly extract a
winning region but also implementations for the control signals from the
$C_j(\overline{X}, \overline{I})$-predicates.  \iprover\ also won the EPR track
of the Automated Theorem Proving System Competition in the last years.

\section{Experimental Results}
\label{sec:exp_res}

This section presents our implementation, benchmarks and experimental results.

\subsection{Implementation}
\label{sec:impl}

We implemented the synthesis methods presented in this paper in a prototype
tool.  The source code (written in C++), more extensive experimental results,
and the scripts to reproduce them are available for
download\footnote{\url{www.iaik.tugraz.at/content/research/design_verification/demiurge/}.}.
Our tool takes as input an \aiger\footnote{See \url{http://fmv.jku.at/aiger/}.}
file, defined as for the safety track of the hardware synthesis competition, but
with the inputs separated into controllable and uncontrollable ones.  It outputs
the synthesized implementation in \aiger format as well.  Several back-ends
implement different methods to compute a winning region.  At the moment, they
all use \qbfcert~\cite{NiemetzPLSB12} to extract the final implementation.
However, in this paper, we evaluate the winning region computation only.
Table~\ref{tab:impl} describes some of our implementations.  Results for more
configurations (with different optimizations, solvers, etc.) can be found in the
downloadable archive.
\begin{table}[tb]
\centering
\caption{Overview of our Implementations}
\label{tab:impl}
\setlength{\tabcolsep}{4.3pt}
\begin{tabular}{l||c|c|l}
Name  & Techn.      & Solver              & Description\\
\hline
BDD   & BDDs       & \cudd                & \textsc{SafeSynth} (Sect.~\ref{sec:std})                         \\
PDM   & SAT        & \minisat             & Property directed method~\cite{MorgensternGS13}                  \\
QAGB  & QBF        & \bloqqerM\ + \depqbf  & \textsc{LearnQbf} + opt. \textsf{RG} + comp. of all  \\
      &            &                      & counterexample generalizations (Sect.~\ref{sec:qbf_learn})       \\
SM    & SAT        & \minisat             & \textsc{LearnSat} (Sect.~\ref{sec:sat_learn})                    \\
SGM   & SAT        & \minisat             & Like SM but with optimization \textsf{RG}                        \\
P$i$  & SAT        & various              & Multi-threaded with $i$ threads (Sect.~\ref{sec:parallel})       \\
TB    & QBF        & \bloqqerM\ + \depqbf  & CNF-template-based (Sect.~\ref{sec:std})                         \\
EPR   & EPR        & \iprover             & EPR-based (Sect.~\ref{sec:epr})
\end{tabular}
\end{table}
The BDD-based method is actually implemented in a separate tool\footnote{Is was
created by students and won a competition in a lecture on synthesis.}.  It uses
dynamic variable reordering, forced re-orderings at certain points, and a cache
to speedup the construction of the transition relation.  PDM is a
re-implementation of~\cite{MorgensternGS13}.  These two implementations serve as
baseline for our comparison.  The other methods are implemented as described
above. \bloqqerM\ refers to an extension of the QBF-preprocessor \bloqqer\ to
preserve satisfying assignments.  This extension is presented
in~\cite{SeidlK14}.


\subsection{Benchmarks}
\label{sec:bench}

We evaluate the methods on several parametrized specifications. The first one
defines an arbiter for ARM's AMBA AHB bus~\cite{BloemGJPPW07b}.  It is
parametrized with the number of masters it can handle.  These specifications are
denoted as \texttt{amba}$ij$, where $i$ is the number of masters, and
$j\in\{\texttt{c},\texttt{b}\}$ indicates how the fairness properties in the
original formulation of the specification were transformed into safety
properties (see Appendix \ref{sec:appendix_bench} for details).  The second
specification is denoted by \texttt{genbuf}$ij$, with
$j\in\{\texttt{c},\texttt{b}\}$, and defines a generalized
buffer~\cite{BloemGJPPW07b} connecting $i$ senders to two receivers.  Also here,
liveness properties have been reduced to safety properties.  Both of these
specifications can be considered as ``control-intensive'', i.e., contain
complicated constraints on few signals.  In contrast to that, the following
specifications are more ``data-intensive'', and do not contain transformed
liveness properties.  The specification \texttt{add}$io$ with
$o\in\{\texttt{y},\texttt{n}\}$ denotes a combinational $i$-bit adder.  Here
$o$=$\texttt{y}$ indicates that the \aiger file was optimized with
\Abc~\cite{BraytonM10}, and $o$=$\texttt{n}$ means that this optimization was
skipped. Next, \texttt{mult}$i$ denotes a combinational $i$-bit multiplier. The
benchmark \texttt{cnt}$io$ denotes an $i$-bit counter that must not reach its
maximum value, which can be prevented by setting the control signals correctly
at some other counter value. Finally, \texttt{bs}$io$ denotes an $i$-bit barrel
shifter that is controlled by some signals.  The tables \ref{table:results1} and
\ref{table:results2} in Appendix \ref{sec:time_tables} list the size of these
benchmarks.

\subsection{Results}
\label{sec:res}

Figure~\ref{fig:cactus} summarizes the performance results of our synthesis
methods on the different parameterized specifications with cactus plots.  The
vertical axis shows the execution time for computing a winning region using a
logarithmic scale.
\begin{figure}[bt]
\centering
   \subfigure[Results for \texttt{amba}\label{fig:cactus_amba}]
             {\includegraphics[width=0.49\textwidth]{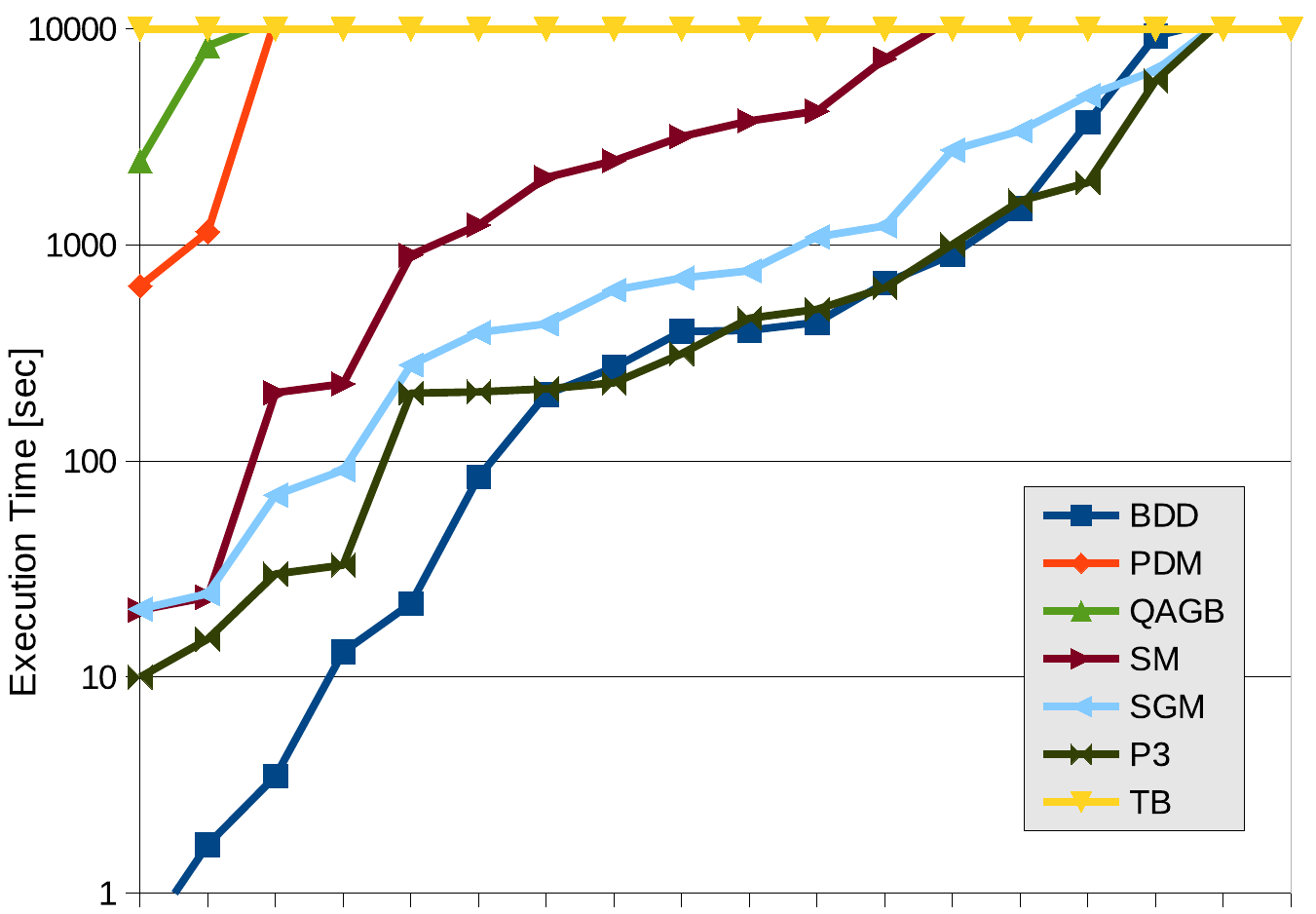}}
   \subfigure[Results for \texttt{genbuf}\label{fig:cactus_genbuf}]
             {\includegraphics[width=0.49\textwidth]{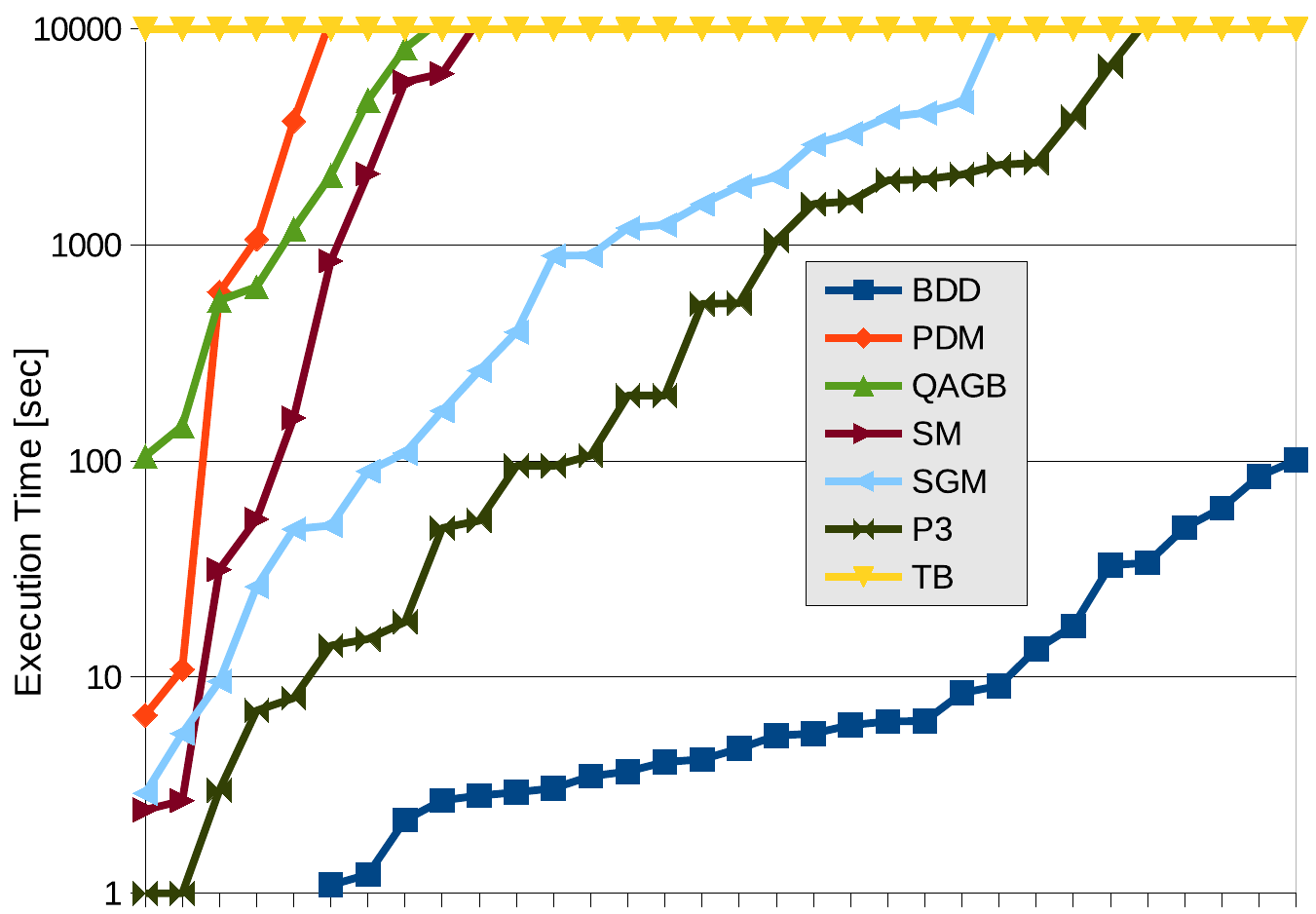}}
   \subfigure[Results for \texttt{add}\label{fig:cactus_add}]
             {\includegraphics[width=0.49\textwidth]{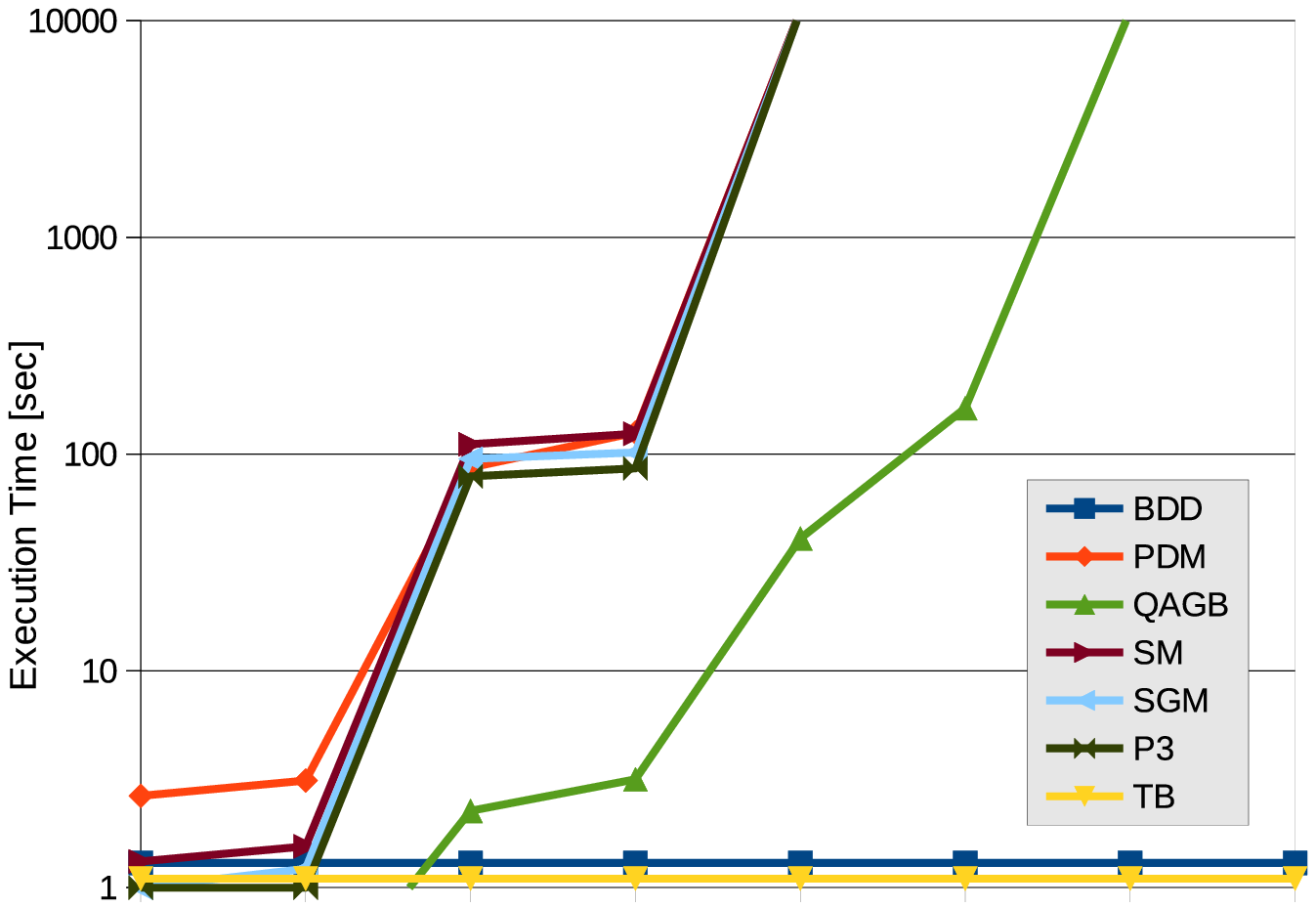}}
   \subfigure[Results for \texttt{mult}\label{fig:cactus_mult}]
             {\includegraphics[width=0.49\textwidth]{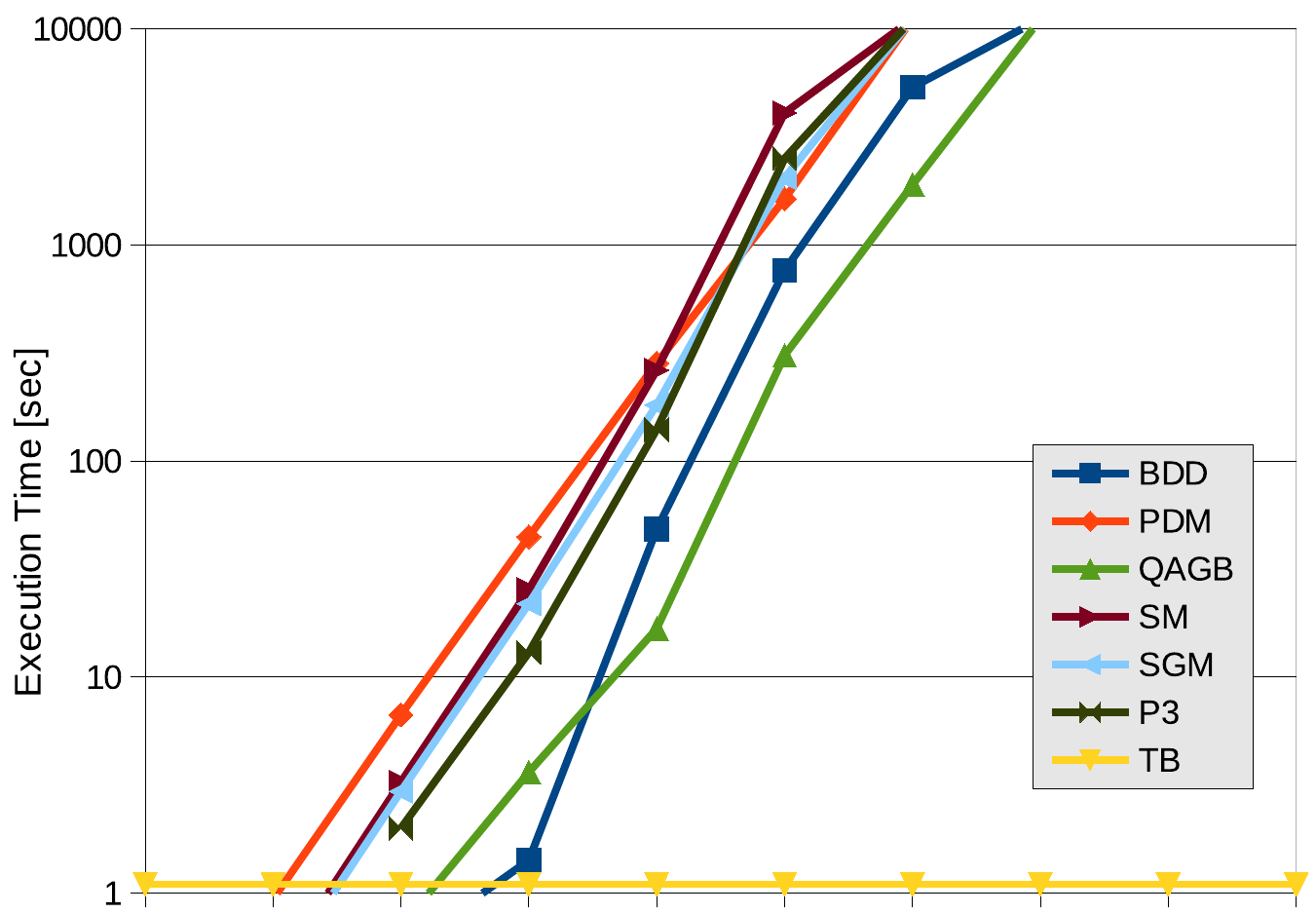}}
   \subfigure[Results for \texttt{cnt}\label{fig:cactus_cnt}]
             {\includegraphics[width=0.49\textwidth]{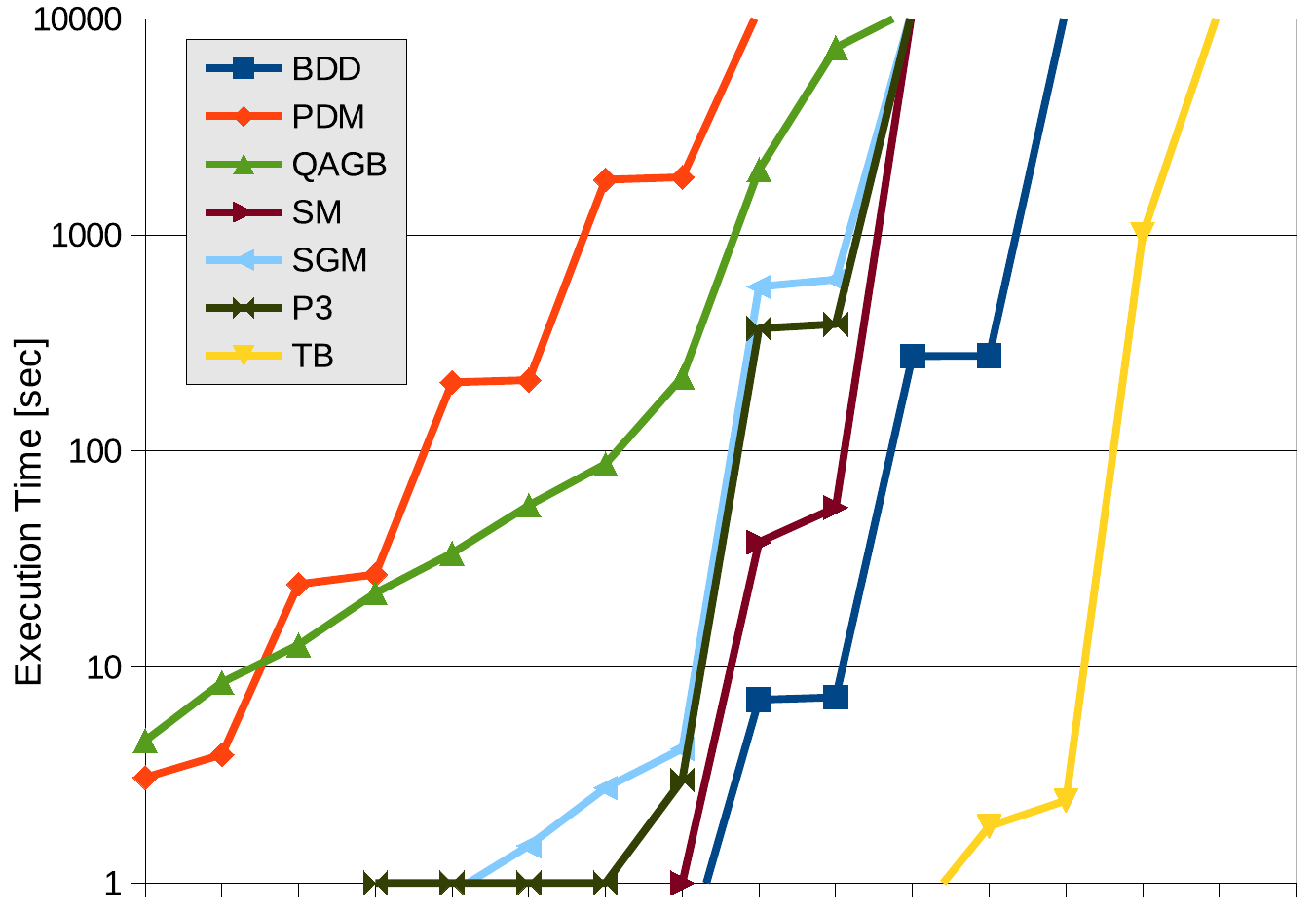}}
   \subfigure[Results for \texttt{bs}\label{fig:cactus_bs}]
             {\includegraphics[width=0.49\textwidth]{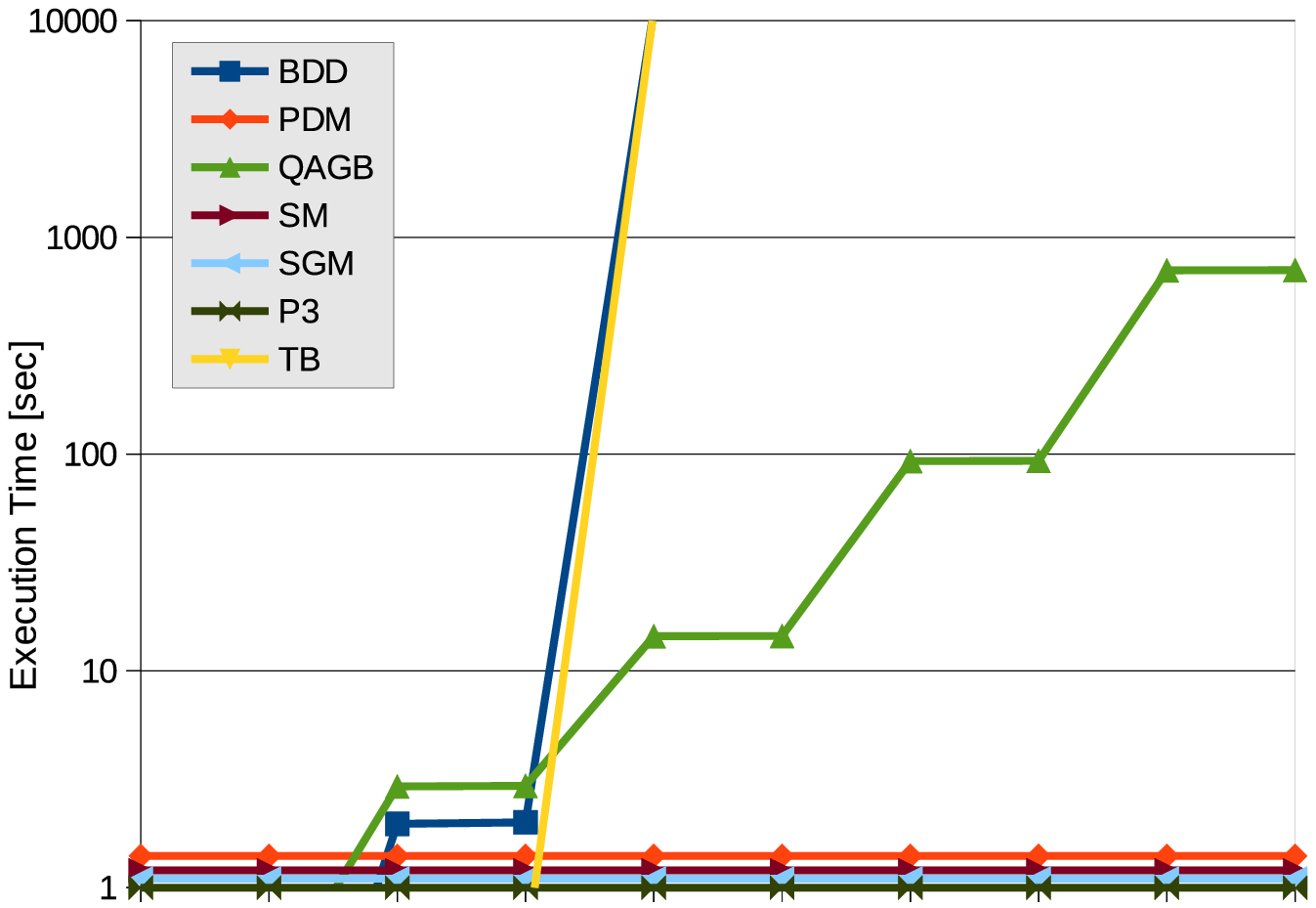}}
   \caption{Cactus plots summarizing our performance evaluation.}
  \label{fig:cactus}
\end{figure}
The horizontal axis gives the number of benchmark instances that can be solved
within this time limit (per instance).  Roughly speaking this means that the
steeper a line rises, the worse is the scalability of this method.  In order to
make the charts more legible, we sometimes ``zoomed'' in on the interesting
parts.  That is, in some charts we omitted the leftmost part were all methods
terminate within fractions of a second, as well as the rightmost part where
(almost) all
methods timeout.  We set a timeout of $10\,000$ seconds, and a
memory limit of $4$\,GB. The memory limit was only exceeded by the EPR approach.
The EPR approach did so for quite small instances already, so we did not include
it in Figure~\ref{fig:cactus}. The detailed execution times can be found in the
tables \ref{table:results1} and \ref{table:results2} of Appendix
\ref{sec:time_tables}. All experiments were performed on an Intel Xeon E5430 CPU
with 4 cores running at $2.66$\,GHz, and a 64 bit Linux.
Figure~\ref{fig:speedup} illustrates the speedup achieved by our parallelization
(see Section~\ref{sec:parallel}) on the \texttt{amba} and \texttt{genbuf}
benchmarks in a scatter plot.  The x-axis carries the computation time with one
thread.  The y-axis shows the corresponding execution time with two and three
threads.  Note that the scale on both axes is logarithmic.

\subsection{Discussion}
\label{sec:disc}

\begin{wrapfigure}[16]{r}{0.5\textwidth}
\vspace{-0.9cm}
 \centering
 \includegraphics[width=0.5\textwidth]{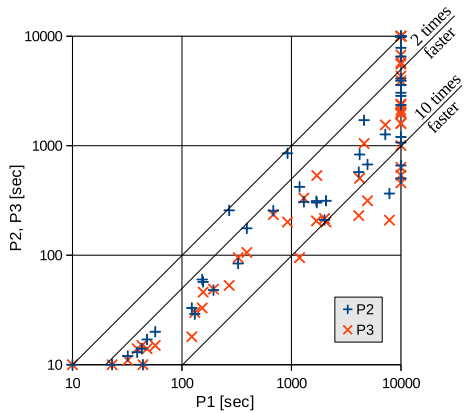}
 \caption{Parallelization speedup.}
 \label{fig:speedup}
\end{wrapfigure}
Figure~\ref{fig:speedup} illustrates a parallelization speedup mostly between a
factor of $2$ and $37$, with a tendency to greater improvements for larger
benchmarks.  Only part of the speedup is due to the exploitation of hardware
parallelism.  Most of the speedup actually stems from the fact that the threads
in our parallelization execute different methods and use different solvers that
complement each other.  Even if executed on a single CPU core in a
pseudo-parallel manner, a significant speedup can be observed.  In our
parallelization, we experimented with only a few combinations of solvers and
algorithms.  We think that there is still a lot of room for improvements,
requiring a more extensive investigation of beneficial algorithm and solver
combinations.

For the \texttt{amba} benchmarks, our parallelization P3 slightly outperforms
BDDs (Figure~\ref{fig:cactus_amba}). For \texttt{genbuf}, BDDs are significantly
faster (Figure~\ref{fig:cactus_genbuf}). The template-based approach does not
scale at all for these benchmarks.  The reason is that, most likely, no simple
CNF representation of a winning region exists for these benchmarks.  For
instance, for the smallest \texttt{genbuf} instance, P3 computes a winning
region as a CNF formula with $124$ clauses and $995$ literal occurrences. By
dropping literals and clauses as long as this does not change the shape of the
winning region, we can simplify this CNF to $111$ clauses and $849$ literal
occurrences.  These numbers indicates that no winning region for these benchmarks
can be described with only a few clauses.  Instantiating a CNF template with
more than $100$ clauses is far beyond the capabilities of the solver, because
the number of template parameters grows so large (e.g., $4300$ template
parameters for the smallest \texttt{genbuf} instance with a template of $100$
clauses for the winning region).  The situation is different for \texttt{add}
and \texttt{mult}. These designs are mostly combinational (with a few states to
track if an error occurred).  A simple CNF-representation of the winning region
(with no more than 2 clauses) exists, and the template-based approach finds it
quickly (Figure~\ref{fig:cactus_add} and~\ref{fig:cactus_mult}).

In Figure~\ref{fig:cactus_genbuf}, we observe a great improvement due to the
reachability optimization \textsf{RG} (SM vs.~SGM). In some plots, this
improvement is not so significant, but optimization \textsf{RG} never slows down
the computation significantly.  Similar observations can be made for QAGB (but
this is not shown in the plots to keep them simple).

The SAT-based back-end SGM outperforms the QBF-based back-end QAGB on most
benchmark classes (all except for \texttt{add} and \texttt{mult}).  It has
already been observed before that solving QBF-problems with plain SAT-solvers
can be beneficial~\cite{JanotaS11,MorgensternGS13}. Our experiments confirm
these observations. One possible reason is that SAT-solvers can be used
incrementally, and they can compute unsatisfiable cores. These features are
missing in modern QBF-solvers. However, this situation may change in the future.

The barrel shifters \texttt{bs} are intractable for BDDs, even for rather small
sizes.  Already when building the BDD for the transition relation, the approach
times out because of many and long reordering phases, or runs out of memory if
reordering is disabled. In contrast, almost all our SAT- and QBF-based
approaches are done within fractions of a second on these examples.  We can
consider the \texttt{bs}-benchmark as an example of a design with complex
data-path elements.  BDDs often fail to represent such elements efficiently. In
contrast, the SAT- and QBF-based methods can represent them easily in CNF. At
the same time, the SAT- and QBF-solvers seem to be smart enough to consider the
complex data-path elements only as far as they are relevant for the synthesis
problem.

On most of the benchmarks, especially \texttt{amba} and \texttt{genbuf}, our new
synthesis methods outperform our re-implementation of~\cite{MorgensternGS13}
(PDM in Figure~\ref{fig:cactus}) by orders of magnitude. Yet,
\cite{MorgensternGS13} reports impressive results for these benchmarks: the
synthesis time is below 10 seconds even for \texttt{amba}$16$ and
\texttt{genbuf}$16$.  We believe that this is due to a different formulation of
the benchmarks.  We translated the benchmarks, exactly as used in
\cite{MorgensternGS13}, into our input language manually, at least for
\texttt{amba}$16$ and \texttt{genbuf}$16$.  Our PDM back-end, as well as most of
the other back-ends, solve them in a second.  This suggests that the enormous
runtime differences stem from differences in the benchmarks, and not in the
implementation. An investigation of the exact differences in the benchmarks
remains to be done.

In summary, none of the approaches is consistently superior.  Instead, the
different benchmark classes favor different methods.  BDDs perform well on many
benchmarks, but are outperformed by our new methods on some classes. The
template-based approach and the parallelization of the SAT-based approach seem
particularly promising.  The reduction to EPR turned out to scale poorly.

\section{Summary and Conclusion}
\label{sec:concl}

In this paper, we presented various novel SAT- and QBF-based methods to
synthesize finite-state systems from safety specifications.  We started with a
learning-based method that can be implemented with a QBF-solver.  Next, we
proposed an efficient implementation using a SAT-solver, an optimization using
reachability information, and an efficient parallelization that achieves a
super-linear speedup by combining different methods and solvers. Complementary
to that, we also presented synthesis methods based on templates or reduction to
EPR.  From our extensive case study, we conclude that these new methods can
complement BDD-based approaches, and outperform other existing
work~\cite{MorgensternGS13} by orders of magnitude.

In the future, we plan to fine-tune our optimizations and heuristics using
larger benchmark sets.  We also plan to research and compare different methods
for the extraction of circuits from the winning region.

\subsection*{Acknowledgments}
We thank Aaron R. Bradley for fruitful discussions about using IC3-concepts in
synthesis, Andreas Morgenstern for his support in
re-implementing~\cite{MorgensternGS13} and translating benchmarks, Bettina
K\"onighofer also for providing benchmarks, and Fabian Tschiatschek and Mario
Werner for their BDD-based synthesis tool.

\bibliography{references_short}

\newpage
\appendix
\section{Utilizing Unreachable States}

\subsection{Proof of Theorem~\ref{th:gen_reach1}}

\label{sec:appendix_proof}

Theorem~\ref{th:gen_reach1} (cf.~Section.~\ref{sec:reach}) can be proven as
follows.

\begin{proof}
By contradiction, assume that there exists such as state $\mathbf{x}_a$.  Any
implementation $\mathcal{I}$ derived from $W$ will only visit states in $W$.
Hence, there must exist a finite prefix $\mathbf{x}_0,\ldots \mathbf{x}_n$ of an
execution of $\mathcal{S}$ with $\mathbf{x}_0 \models I$, $\mathbf{x}_n =
\mathbf{x}_a$, $n\geq 1$ ($\mathbf{x}_a$ cannot be initial because this would
satisfy Eq.~\ref{eq:gen_reach}), and $\mathbf{x}_j \models W$ for all $0\leq
j\leq n$.  Such a trace is illustrated in Figure~\ref{fig:gen_reach_proof}.
Since $\mathbf{x}_a \models G \wedge \mathbf{x}_g$, there must exist a smallest
$k\leq n$ such that $\mathbf{x}_j \models G\wedge \mathbf{x}_g$ for all $k \leq
j \leq n$.  Now, $\mathbf{x}_k$ is either initial or has a predecessor
$\mathbf{x}_{k-1}$ in $G \wedge \neg \mathbf{x}_g$ (because
$\mathbf{x}_{k-1}\models W, W\Rightarrow G$, and $\mathbf{x}_{k-1}\not\models
G\wedge \mathbf{x}_g$). Since Eq.~\ref{eq:gen_reach} is unsatisfiable, the
protagonist cannot enforce to go from $\mathbf{x}_k$ to $G$. Hence,
$\mathbf{x}_k$ is not winning for the protagonist and cannot be part of $W$.
This contradiction means that such a path of reachable states ending in
$\mathbf{x}_a$ cannot exist if Eq.~\ref{eq:gen_reach} is unsatisfiable.
\end{proof}

\begin{figure}[hb]
\centering
 \centering
 \includegraphics[width=0.3\textwidth]{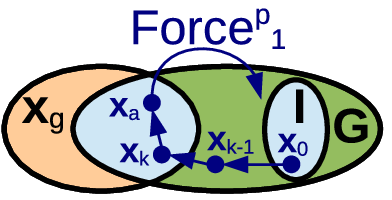}
 \caption{Optimization \textsf{RG}: Proof Illustration.}
 \label{fig:gen_reach_proof}
\end{figure}

\subsection{Optimization \textsf{RC}.}

\label{sec:appendix_rc}

For \textsc{LearnQbf} to find only counterexample-states that could be reachable
in the implementation $\mathcal{I}$, we modify the QBF-call in
line~\ref{alg:learn_ce} to
\begin{align}
(\mathsf{sat},\mathbf{x}):=\qbfsatmodel(
&
      \exists \overline{x}^*, \overline{i}^*, \overline{c}^*\scope
      \exists \overline{x},\overline{i} \scope
      \forall \overline{c} \scope
      \exists \overline{x}' \scope
\notag\\&
(I(\overline{x}) \vee (\overline{x}^*\neq \overline{x}) \wedge
F(\overline{x}^*)
\wedge T(\overline{x}^*, \overline{i}^*,\overline{c}^*, \overline{x}))
\wedge \quad
\label{eq:comp_reach}
\\&
F(\overline{x}) \wedge
T(\overline{x},\overline{i},\overline{c},\overline{x}') \wedge \neg
F(\overline{x}').\notag
\end{align}
Only the second line of the formula is new.  Here, $\overline{x}^*$,
$\overline{i}^*$, and $\overline{c}^*$ are the previous-state copies of
$\overline{x}$, $\overline{i}$, and $\overline{c}$, respectively.  The
additional constraint requires that the counterexample-state $\mathbf{x}$ is
either initial, or it has a predecessor in $F$ that is different from
$\mathbf{x}$.  Otherwise it must be unreachable and can be ignored. A state
$\mathbf{x}\not\models I$ that has itself as the only predecessor in $F$ is, of
course, unreachable as well.  When using optimization \textsf{RC}, the resulting
winning region $W$ of \textsc{LearnQbf} may not satisfy condition (III), i.e.,
$W \Rightarrow \FS{1}(W)$, but only the negation of Eq.~\ref{eq:comp_reach}.
Still, a safe controller can easily be extracted, e.g., by computing Skolem
functions for the $\overline{c}$-signals in this negation of
Eq.~\ref{eq:comp_reach}.

In case of unrealizability, optimization \textsf{RC} cannot make an unrealizable
specification be identified as realizable.  This also holds in combination with
optimization \textsf{RG} and is argued by the following theorem.
\begin{theorem}
For an unrealizable specification, when using optimization \textsf{RC},
\textsc{LearnQbf} will always find Eq.~\ref{eq:comp_reach} satisfiable.
\end{theorem}
\begin{proof}
By contradiction, assume that Eq.~\ref{eq:comp_reach} is unsatisfiable.  We have
that $I\Rightarrow F$ (otherwise \textsc{LearnQbf} would already have terminated
signaling unrealizability) and $F\Rightarrow P$.  Since the specification is
unrealizable, there must exist a state $\mathbf{x}_a\models F$ which is
reachable from within $F$ and from which the antagonist can enforce to leave
$F$.  That is, $\exists \overline{x},\overline{i} \scope \forall \overline{c}
\scope \exists \overline{x}' \scope \mathbf{x}_a \wedge T \wedge \neg F'$, and
there exists a prefix $\mathbf{x}_0,\ldots \mathbf{x}_n$ of an execution of
$\mathcal{S}$ with $\mathbf{x}_0 \models I$,  $\mathbf{x}_n = \mathbf{x}_a$, and
$\mathbf{x}_j \models F$ for all $0\leq j\leq n$.  If $\mathbf{x}_a\models I$,
then Eq.~\ref{eq:comp_reach} is satisfied.  If $\mathbf{x}_a\not\models I$, then
there must exist a maximum $k$ such that $\mathbf{x}_k \neq \mathbf{x}_a$. Since
$\mathbf{x}_k \models F$, this would also satisfy Eq.~\ref{eq:comp_reach}. This
contradiction implies that \textsc{LearnQbf} cannot find Eq.~\ref{eq:comp_reach}
unsatisfiable in case of unrealizability.
\end{proof}

\section{More Detailed Experimental Results}

\subsection{Benchmark Creation}

\label{sec:appendix_bench}

Most of the benchmarks were created as follows.  First, we created a declarative
system description in Verilog (i.e., stated which behavior is allowed/not
allowed). Second, the Verilog file was translated into the BLIF-MV format using
\vltomv\footnote{\url{http://vlsi.colorado.edu/~vis/}}. Third, the BLIF-MV file
was translated into \aiger format using \Abc~\cite{BraytonM10}.

The \texttt{amba} and \texttt{genbuf} benchmarks are translations of \ratsy's
input files\footnote{\url{http://rat.fbk.eu/ratsy/}} into \aiger using the flow
described above.  \ratsy takes as input specifications in so-called
``Generalized Reactivity(1)'' format.  Such specifications consist of two parts:
assumptions and guarantees.  Both parts consist of safety and fairness
properties.  We used two different methods to reduce these Generalized
Reactivity(1) specifications into pure safety specification: $j=c$ in the
benchmark name \texttt{amba}$ij$ or \texttt{genbuf}$ij$ means that all fairness
assumptions are compressed into one fairness assumption $X$ using a counting
construction.  The same is done for the fairness guarantees.  Finally, an
additional counter is introduced. It is incremented whenever $X$ is satisfied,
it is reset whenever the counting construction for the fairness guarantees
switches to the next guarantee, and it must never reach a given value $N$.  This
enforces a certain ratio between the progress in the fairness assumptions and
guarantees.  We set $N$ to be the minimal value for which the resulting safety
specification is realizable. The value $j=b$ in the benchmark name refers to the
same construction, but using a special counting construction.  It has one bit
per fairness assumption and guarantee. This bit tracks if the property has
already been satisfied or not. Hence, $j=b$ allows the implementation to satisfy
guarantees in arbitrary order, while $j=c$ enforces a certain order between the
fairness properties, but uses less state bits.

\subsection{More Performance Results}

\label{sec:time_tables}
Table~\ref{table:results1} and Table~\ref{table:results2} summarize the size of
the benchmarks, as well as the time needed by the different synthesis methods to
compute the winning region.  The circuit extraction time is not included.  The
suffix ``k'' stands for a multiplication of the respective number by $1000$. The
entries ``$>$10k'' mark a time-out with a limit of $10\,000$ seconds.  The
entries ``$>$4GB'' indicate that the memory limit of $4$\,GB was exceeded. More
detailed performance data (more different configurations of the implementations,
more benchmarks, and more detailed statistics like numbers of iterations,
solving times for different kinds of queries, etc.) can be found in the
downloadable archive\footnote{\url{
www.iaik.tugraz.at/content/research/design_verification/demiurge/}}.

\begin{table*}[ht]
\setlength{\tabcolsep}{2.8pt}
\caption{More extensive performance results, part 1.}
\label{table:results1}
\centering
\scriptsize
\begin{tabular}{lccccccccccccccccccccccccccccccccc}
\toprule[1.3pt]
                     &\multicolumn{4}{c}{Size}
                     &&\multicolumn{10}{c}{Execution Time}
                     \\
                     &$|\overline{i}|$
                     &$|\overline{c}|$
                     &$|\overline{x}|$
                     &$G$
                     &
                     &BDD
                     &PDM
                     &QAGB
                     &SM
                     &SGM
                     &P1
                     &P2
                     &P3
                     &TB
                     &EPR
                     \\
\cmidrule{2-5}
\cmidrule{7-16}
                    &[-]    &[-]  &[-]    &[-]
                    &
                    &[sec]
                    &[sec]
                    &[sec]
                    &[sec]
                    &[sec]
                    &[sec]
                    &[sec]
                    &[sec]
                    &[sec]
                    &[sec]
\\
\cmidrule{1-5}
\cmidrule{7-16}
 \texttt{add2n}     & 4  & 2   & 2   & 23    && 0.1    & 0.1    & 0.1    & 0.1    & 0.1    & 1      & 1      & 1      & 0.1     & 1.0    \\
 \texttt{add2y}     & 4  & 2   & 2   & 17    && 0.1    & 0.1    & 0.1    & 0.1    & 0.1    & 1      & 1      & 1      & 0.1     & 0.1    \\
 \texttt{add4n}     & 8  & 4   & 2   & 61    && 0.1    & 0.1    & 0.1    & 0.1    & 0.1    & 1      & 1      & 1      & 0.1     & 125    \\
 \texttt{add4y}     & 8  & 4   & 2   & 45    && 0.1    & 0.1    & 0.1    & 0.1    & 0.1    & 1      & 1      & 1      & 0.1     & 120    \\
 \texttt{add6n}     & 12 & 6   & 2   & 99    && 0.1    & 3.1    & 0.2    & 1.6    & 1.2    & 2      & 1      & 1      & 0.1     & $>$4GB \\
 \texttt{add6y}     & 12 & 6   & 2   & 73    && 0.1    & 2.7    & 0.3    & 1.3    & 1.0    & 3      & 1      & 1      & 0.1     & $>$4GB \\
 \texttt{add8n}     & 16 & 8   & 2   & 137   && 0.1    & 126    & 3.2    & 111    & 102    & 329    & 71     & 86     & 0.1     & $>$4GB \\
 \texttt{add8y}     & 16 & 8   & 2   & 101   && 0.1    & 87     & 2.3    & 124    & 95     & 318    & 81     & 79     & 0.1     & $>$4GB \\
 \texttt{add10n}    & 20 & 10  & 2   & 175   && 0.1    & $>$10k & 163    & $>$10k & $>$10k & $>$10k & $>$10k & $>$10k & 0.1     & $>$4GB \\
 \texttt{add10y}    & 20 & 10  & 2   & 129   && 0.1    & $>$10k & 41     & $>$10k & $>$10k & $>$10k & $>$10k & $>$10k & 0.1     & $>$4GB \\
 \texttt{add12n}    & 24 & 12  & 2   & 213   && 0.1    & $>$10k & $>$10k & $>$10k & $>$10k & $>$10k & $>$10k & $>$10k & 0.1     & $>$4GB \\
 \texttt{add12y}    & 24 & 12  & 2   & 157   && 0.1    & $>$10k & $>$10k & $>$10k & $>$10k & $>$10k & $>$10k & $>$10k & 0.1     & $>$4GB \\
 \texttt{add14n}    & 28 & 14  & 2   & 251   && 0.1    & $>$10k & $>$10k & $>$10k & $>$10k & $>$10k & $>$10k & $>$10k & 0.1     & $>$4GB \\
 \texttt{add14y}    & 28 & 14  & 2   & 185   && 0.1    & $>$10k & $>$10k & $>$10k & $>$10k & $>$10k & $>$10k & $>$10k & 0.1     & $>$4GB \\
 \texttt{add16n}    & 32 & 16  & 2   & 289   && 0.1    & $>$10k & $>$10k & $>$10k & $>$10k & $>$10k & $>$10k & $>$10k & 0.1     & $>$4GB \\
 \texttt{add16y}    & 32 & 16  & 2   & 213   && 0.1    & $>$10k & $>$10k & $>$10k & $>$10k & $>$10k & $>$10k & $>$10k & 0.1     & $>$4GB \\
 \texttt{add18n}    & 36 & 18  & 2   & 327   && 0.1    & $>$10k & $>$10k & $>$10k & $>$10k & $>$10k & $>$10k & $>$10k & 0.1     & $>$4GB \\
 \texttt{add18y}    & 36 & 18  & 2   & 241   && 0.1    & $>$10k & $>$10k & $>$10k & $>$10k & $>$10k & $>$10k & $>$10k & 0.1     & $>$4GB \\
 \texttt{add20n}    & 40 & 20  & 2   & 365   && 0.1    & $>$10k & $>$10k & $>$10k & $>$10k & $>$10k & $>$10k & $>$10k & 0.1     & $>$4GB \\
 \texttt{add20y}    & 40 & 20  & 2   & 269   && 0.1    & $>$10k & $>$10k & $>$10k & $>$10k & $>$10k & $>$10k & $>$10k & 0.1     & $>$4GB \\
\cmidrule{1-5}
\cmidrule{7-16}
 \texttt{cnt4n}     & 1  & 1   & 5   & 60    && 0.1    & 0.1    & 0.1    & 0.1    & 0.1    & 1      & 1      & 1      & 0.1     & 5.0    \\
 \texttt{cnt4y}     & 1  & 1   & 5   & 23    && 0.1    & 0.1    & 0.1    & 0.1    & 0.1    & 1      & 1      & 1      & 0.1     & 2.0    \\
 \texttt{cnt5n}     & 1  & 1   & 6   & 75    && 0.1    & 0.1    & 0.5    & 0.1    & 0.1    & 1      & 1      & 1      & 0.2     & 40     \\
 \texttt{cnt5y}     & 1  & 1   & 6   & 29    && 0.1    & 0.1    & 0.1    & 0.1    & 0.1    & 1      & 1      & 1      & 0.2     & 14     \\
 \texttt{cnt6n}     & 1  & 1   & 7   & 90    && 0.1    & 0.1    & 1.6    & 0.1    & 0.1    & 1      & 1      & 1      & 0.1     & $>$4GB \\
 \texttt{cnt6y}     & 1  & 1   & 7   & 35    && 0.1    & 0.1    & 0.4    & 0.1    & 0.1    & 1      & 1      & 1      & 0.1     & 81     \\
 \texttt{cnt7n}     & 1  & 1   & 8   & 105   && 0.1    & 0.6    & 4.6    & 0.1    & 0.1    & 1      & 1      & 1      & 0.1     & $>$4GB \\
 \texttt{cnt7y}     & 1  & 1   & 8   & 41    && 0.1    & 0.4    & 1.2    & 0.1    & 0.1    & 1      & 1      & 1      & 0.1     & $>$4GB \\
 \texttt{cnt8n}     & 1  & 1   & 9   & 120   && 0.1    & 3.9    & 13     & 0.1    & 0.2    & 1      & 1      & 1      & 0.2     & $>$4GB \\
 \texttt{cnt8y}     & 1  & 1   & 9   & 47    && 0.1    & 3.1    & 3.2    & 0.1    & 0.1    & 1      & 1      & 1      & 0.2     & $>$4GB \\
 \texttt{cnt9n}     & 1  & 1   & 10  & 135   && 0.1    & 27     & 34     & 0.2    & 0.6    & 1      & 1      & 1      & 0.2     & $>$4GB \\
 \texttt{cnt9y}     & 1  & 1   & 10  & 53    && 0.1    & 24     & 8.5    & 0.1    & 0.3    & 1      & 1      & 1      & 0.2     & $>$4GB \\
 \texttt{cnt10n}    & 1  & 1   & 11  & 150   && 0.2    & 213    & 87     & 0.4    & 1.5    & 1      & 1      & 1      & 0.2     & $>$4GB \\
 \texttt{cnt10y}    & 1  & 1   & 11  & 59    && 0.2    & 208    & 22     & 0.2    & 0.9    & 1      & 1      & 1      & 0.2     & $>$4GB \\
 \texttt{cnt11n}    & 1  & 1   & 12  & 165   && 0.4    & 1.8k   & 220    & 1.0    & 4.2    & 2      & 3      & 3      & 0.6     & $>$4GB \\
 \texttt{cnt11y}    & 1  & 1   & 12  & 65    && 0.4    & 1.8k   & 56     & 0.5    & 2.8    & 1      & 1      & 1      & 0.3     & $>$4GB \\
 \texttt{cnt15n}    & 1  & 1   & 16  & 225   && 7.2    & $>$10k & 7.4k   & 55     & 621    & 232    & 554    & 369    & 0.4     & $>$4GB \\
 \texttt{cnt15y}    & 1  & 1   & 16  & 89    && 7.1    & $>$10k & 2.0k   & 38     & 576    & 205    & 620    & 386    & 0.4     & $>$4GB \\
 \texttt{cnt20n}    & 1  & 1   & 21  & 300   && 276    & $>$10k & $>$10k & $>$10k & $>$10k & $>$10k & $>$10k & $>$10k & 2.4     & $>$4GB \\
 \texttt{cnt20y}    & 1  & 1   & 21  & 119   && 275    & $>$10k & $>$10k & $>$10k & $>$10k & $>$10k & $>$10k & $>$10k & 1.8     & $>$4GB \\
 \texttt{cnt25n}    & 1  & 1   & 26  & 375   && $>$10k & $>$10k & $>$10k & $>$10k & $>$10k & $>$10k & $>$10k & $>$10k & 1.0k    & $>$4GB \\
 \texttt{cnt25y}    & 1  & 1   & 26  & 149   && $>$10k & $>$10k & $>$10k & $>$10k & $>$10k & $>$10k & $>$10k & $>$10k & $>$10k  & $>$4GB \\
 \texttt{cnt30n}    & 1  & 1   & 31  & 450   && $>$10k & $>$10k & $>$10k & $>$10k & $>$10k & $>$10k & $>$10k & $>$10k & $>$10k  & $>$4GB \\
 \texttt{cnt30y}    & 1  & 1   & 31  & 179   && $>$10k & $>$10k & $>$10k & $>$10k & $>$10k & $>$10k & $>$10k & $>$10k & 0.7     & $>$4GB \\
\cmidrule{1-5}
\cmidrule{7-16}
 \texttt{mult2}     & 4  & 4   & 0   & 24    && 0.1    & 0.1    & 0.1    & 0.1    & 0.1    & 1      & 1      & 1      & 0.1     & 0.1    \\
 \texttt{mult4}     & 8  & 8   & 0   & 128   && 0.1    & 0.1    & 0.1    & 0.1    & 0.1    & 1      & 1      & 1      & 0.1     & 30     \\
 \texttt{mult5}     & 10 & 10  & 0   & 217   && 0.1    & 0.9    & 0.3    & 0.4    & 0.4    & 1      & 1      & 1      & 0.1     & $>$4GB \\
 \texttt{mult6}     & 12 & 12  & 0   & 322   && 0.5    & 6.7    & 0.7    & 3.3    & 3.0    & 5      & 2      & 2      & 0.1     & $>$4GB \\
 \texttt{mult7}     & 14 & 14  & 0   & 455   && 1.4    & 45     & 3.6    & 25     & 22     & 33     & 16     & 13     & 0.1     & $>$4GB \\
 \texttt{mult8}     & 16 & 16  & 0   & 604   && 48     & 284    & 17     & 264    & 182    & 519    & 172    & 140    & 0.1     & $>$4GB \\
 \texttt{mult9}     & 18 & 18  & 0   & 759   && 762    & 1.6k   & 309    & 4.1k   & 2.0k   & $>$10k & 2.7k   & 2.5k   & 0.1     & $>$4GB \\
 \texttt{mult10}    & 20 & 20  & 0   & 964   && 5.4k   & $>$10k & 1.9k   & $>$10k & $>$10k & $>$10k & $>$10k & $>$10k & 0.1     & $>$4GB \\
 \texttt{mult11}    & 22 & 22  & 0   & 1.1k  && $>$10k & $>$10k & $>$10k & $>$10k & $>$10k & $>$10k & $>$10k & $>$10k & 0.1     & $>$4GB \\
 \texttt{mult12}    & 24 & 24  & 0   & 1.4k  && $>$10k & $>$10k & $>$10k & $>$10k & $>$10k & $>$10k & $>$10k & $>$10k & 0.1     & $>$4GB \\
 \texttt{mult13}    & 26 & 26  & 0   & 1.5k  && $>$10k & $>$10k & $>$10k & $>$10k & $>$10k & $>$10k & $>$10k & $>$10k & 0.1     & $>$4GB \\
 \texttt{mult14}    & 28 & 28  & 0   & 1.8k  && $>$10k & $>$10k & $>$10k & $>$10k & $>$10k & $>$10k & $>$10k & $>$10k & 0.1     & $>$4GB \\
 \texttt{mult15}    & 30 & 30  & 0   & 2.0k  && $>$10k & $>$10k & $>$10k & $>$10k & $>$10k & $>$10k & $>$10k & $>$10k & 0.1     & $>$4GB \\
 \texttt{mult16}    & 32 & 32  & 0   & 2.5k  && $>$10k & $>$10k & $>$10k & $>$10k & $>$10k & $>$10k & $>$10k & $>$10k & 0.3     & $>$4GB \\
\bottomrule[1.3pt]
\end{tabular}
\end{table*}

\begin{table*}[ht]
\centering
\setlength{\tabcolsep}{2.4pt}
\caption{More extensive performance results, part 2.}
\label{table:results2}
\scriptsize
\begin{tabular}{lccccccccccccccccccccccccccccccccc}
\toprule[1.3pt]
                     &\multicolumn{4}{c}{Size}
                     &\multicolumn{10}{c}{Execution Time}
                     \\
                     &$|\overline{i}|$
                     &$|\overline{c}|$
                     &$|\overline{x}|$
                     &$G$
                     &
                     &BDD
                     &PDM
                     &QAGB
                     &SM
                     &SGM
                     &P1
                     &P2
                     &P3
                     &TB
                     &EPR
                     \\
\cmidrule{2-5}
\cmidrule{7-16}
                    &[-]    &[-]  &[-]    &[-]
                    &
                    &[sec]
                    &[sec]
                    &[sec]
                    &[sec]
                    &[sec]
                    &[sec]
                    &[sec]
                    &[sec]
                    &[sec]
                    &[sec]
\\
\cmidrule{1-5}
\cmidrule{7-16}
 \texttt{bs08n}     & 2  & 1   & 9   &  82   && 0.1    & 0.1    & 0.3    & 0.1    & 0.1    & 1      & 1      & 1      & 0.1     & $>$4GB \\
 \texttt{bs08y}     & 2  & 1   & 9   &  80   && 0.1    & 0.1    & 0.3    & 0.1    & 0.1    & 1      & 1      & 1      & 0.1     & $>$4GB \\
 \texttt{bs16n}     & 4  & 1   & 17  &  258  && 2.0    & 0.1    & 3.0    & 0.1    & 0.1    & 1      & 1      & 1      & 0.5     & $>$4GB \\
 \texttt{bs16y}     & 4  & 1   & 17  &  256  && 2.0    & 0.1    & 2.9    & 0.1    & 0.1    & 1      & 1      & 1      & 0.4     & $>$4GB \\
 \texttt{bs32n}     & 5  & 1   & 33  &  610  && $>$10k & 0.1    & 15     & 0.1    & 0.1    & 1      & 1      & 1      & $>$10k  & $>$4GB \\
 \texttt{bs32y}     & 5  & 1   & 33  &  608  && $>$10k & 0.1    & 14     & 0.1    & 0.1    & 1      & 1      & 1      & $>$10k  & $>$4GB \\
 \texttt{bs64n}     & 6  & 1   & 65  &  1.4k && $>$10k & 0.1    & 93     & 0.1    & 0.1    & 1      & 1      & 1      & $>$10k  & $>$4GB \\
 \texttt{bs64y}     & 6  & 1   & 65  &  1.4k && $>$10k & 0.1    & 93     & 0.1    & 0.1    & 1      & 1      & 1      & $>$10k  & $>$4GB \\
 \texttt{bs128n}    & 7  & 1   & 129 &  3.2k && $>$10k & 0.1    & 707    & 0.1    & 0.1    & 1      & 1      & 1      & $>$10k  & $>$4GB \\
 \texttt{bs128y}    & 7  & 1   & 129 &  3.2k && $>$10k & 0.1    & 706    & 0.1    & 0.1    & 1      & 1      & 1      & $>$10k  & $>$4GB \\
\cmidrule{1-5}
\cmidrule{7-16}
 \texttt{genbuf01c} & 5  & 6   & 21  &  134  && 0.2    & 11     & 105    & 2.4    & 5.5    & 2      & 1      & 1      & $>$10k  & $>$4GB \\
 \texttt{genbuf02c} & 6  & 7   & 24  &  169  && 0.4    & 605    & 555    & 32     & 50     & 16     & 9      & 7      & $>$10k  & $>$4GB \\
 \texttt{genbuf03c} & 7  & 9   & 27  &  202  && 0.7    & 3.8k   & 2.1k   & 159    & 109    & 43     & 14     & 15     & $>$10k  & $>$4GB \\
 \texttt{genbuf04c} & 8  & 10  & 30  &  242  && 2.9    & $>$10k & 4.7k   & 2.1k   & 892    & 270    & 257    & 53     & $>$10k  & $>$4GB \\
 \texttt{genbuf05c} & 9  & 12  & 33  &  284  && 2.7    & $>$10k & $>$10k & 5.7k   & 1.5k   & 1.2k   & 421    & 95     & $>$10k  & $>$4GB \\
 \texttt{genbuf06c} & 10 & 13  & 35  &  323  && 2.2    & $>$10k & $>$10k & $>$10k & 3.9k   & 921    & 852    & 201    & $>$10k  & $>$4GB \\
 \texttt{genbuf07c} & 11 & 14  & 37  &  361  && 6.2    & $>$10k & $>$10k & $>$10k & 3.3k   & 2.1k   & 314    & 201    & $>$10k  & $>$4GB \\
 \texttt{genbuf08c} & 12 & 15  & 40  &  406  && 8.4    & $>$10k & $>$10k & $>$10k & $>$10k & $>$10k & 6.5k   & 2.4k   & $>$10k  & $>$4GB \\
 \texttt{genbuf09c} & 13 & 17  & 43  &  463  && 13     & $>$10k & $>$10k & $>$10k & $>$10k & $>$10k & $>$10k & 2.0k   & $>$10k  & $>$4GB \\
 \texttt{genbuf10c} & 14 & 18  & 45  &  494  && 9.1    & $>$10k & $>$10k & $>$10k & $>$10k & $>$10k & $>$10k & 2.0k   & $>$10k  & $>$4GB \\
 \texttt{genbuf11c} & 15 & 19  & 47  &  531  && 33     & $>$10k & $>$10k & $>$10k & $>$10k & $>$10k & $>$10k & 2.3k   & $>$10k  & $>$4GB \\
 \texttt{genbuf12c} & 16 & 20  & 49  &  561  && 34     & $>$10k & $>$10k & $>$10k & $>$10k & $>$10k & $>$10k & $>$10k & $>$10k  & $>$4GB \\
 \texttt{genbuf13c} & 17 & 21  & 51  &  602  && 49     & $>$10k & $>$10k & $>$10k & $>$10k & $>$10k & $>$10k & $>$10k & $>$10k  & $>$4GB \\
 \texttt{genbuf14c} & 18 & 22  & 53  &  639  && 101    & $>$10k & $>$10k & $>$10k & $>$10k & $>$10k & $>$10k & $>$10k & $>$10k  & $>$4GB \\
\cmidrule{1-5}
\cmidrule{7-16}
 \texttt{genbuf01b} & 5  & 6   & 23  &  141  && 0.2    & 6.7    & 146    & 2.7    & 2.9    & 2      & 1      & 1      & $>$10k  & $>$4GB \\
 \texttt{genbuf02b} & 6  & 7   & 26  &  174  && 0.5    & 1.1k   & 640    & 54     & 9.6    & 7      & 3      & 3      & $>$10k  & $>$4GB \\
 \texttt{genbuf03b} & 7  & 9   & 30  &  208  && 1.2    & $>$10k & 1.2k   & 845    & 26     & 23     & 8      & 8      & $>$10k  & $>$4GB \\
 \texttt{genbuf04b} & 8  & 10  & 33  &  245  && 1.1    & $>$10k & 8.1k   & 6.2k   & 48     & 48     & 17     & 14     & $>$10k  & $>$4GB \\
 \texttt{genbuf05b} & 9  & 12  & 37  &  282  && 4.7    & $>$10k & $>$10k & $>$10k & 90     & 123    & 33     & 18     & $>$10k  & $>$4GB \\
 \texttt{genbuf06b} & 10 & 13  & 40  &  322  && 3.1    & $>$10k & $>$10k & $>$10k & 171    & 194    & 48     & 49     & $>$10k  & $>$4GB \\
 \texttt{genbuf07b} & 11 & 14  & 43  &  358  && 3.5    & $>$10k & $>$10k & $>$10k & 263    & 326    & 84     & 95     & $>$10k  & $>$4GB \\
 \texttt{genbuf08b} & 12 & 15  & 46  &  395  && 2.8    & $>$10k & $>$10k & $>$10k & 396    & 391    & 176    & 106    & $>$10k  & $>$4GB \\
 \texttt{genbuf09b} & 13 & 17  & 50  &  443  && 4.1    & $>$10k & $>$10k & $>$10k & 895    & 1.7k   & 302    & 534    & $>$10k  & $>$4GB \\
 \texttt{genbuf10b} & 14 & 18  & 53  &  475  && 17     & $>$10k & $>$10k & $>$10k & 1.2k   & $>$10k & 660    & 538    & $>$10k  & $>$4GB \\
 \texttt{genbuf11b} & 15 & 19  & 56  &  510  && 6.3    & $>$10k & $>$10k & $>$10k & 1.2k   & 7.2k   & 1.3k   & 1.5k   & $>$10k  & $>$4GB \\
 \texttt{genbuf12b} & 16 & 20  & 59  &  547  && 3.7    & $>$10k & $>$10k & $>$10k & 2.1k   & 4.6k   & 1.7k   & 1.0k   & $>$10k  & $>$4GB \\
 \texttt{genbuf13b} & 17 & 21  & 62  &  582  && 4.1    & $>$10k & $>$10k & $>$10k & 1.9k   & $>$10k & 2.4k   & 1.6k   & $>$10k  & $>$4GB \\
 \texttt{genbuf14b} & 18 & 22  & 65  &  617  && 6.0    & $>$10k & $>$10k & $>$10k & 2.9k   & $>$10k & 4.1k   & 2.1k   & $>$10k  & $>$4GB \\
 \texttt{genbuf15b} & 19 & 23  & 68  &  652  && 5.4    & $>$10k & $>$10k & $>$10k & 4.1k   & $>$10k & 9.8k   & 6.7k   & $>$10k  & $>$4GB \\
 \texttt{genbuf16b} & 20 & 24  & 71  &  687  && 5.5    & $>$10k & $>$10k & $>$10k & 4.6k   & $>$10k & 7.8k   & 3.9k   & $>$10k  & $>$4GB \\
\cmidrule{1-5}
\cmidrule{7-16}
 \texttt{amba02c}   & 7  & 8   & 28  &  177  && 0.6    & 647    & 2.4k   & 20     & 21     & 44     & 10     & 10     & $>$10k  & $>$4GB \\
 \texttt{amba03c}   & 9  & 10  & 34  &  237  && 3.5    & $>$10k & $>$10k & 228    & 91     & 153    & 60     & 33     & $>$10k  & $>$4GB \\
 \texttt{amba04c}   & 11 & 11  & 38  &  279  && 22     & $>$10k & $>$10k & 898    & 619    & 1.7k   & 312    & 206    & $>$10k  & $>$4GB \\
 \texttt{amba05c}   & 13 & 13  & 43  &  345  && 439    & $>$10k & $>$10k & 1.2k   & 431    & 4.1k   & 574    & 230    & $>$10k  & $>$4GB \\
 \texttt{amba06c}   & 15 & 14  & 47  &  395  && 204    & $>$10k & $>$10k & 2.1k   & 704    & 5.0k   & 674    & 314    & $>$10k  & $>$4GB \\
 \texttt{amba07c}   & 17 & 15  & 52  &  449  && 397    & $>$10k & $>$10k & 4.2k   & 1.2k   & $>$10k & 1.1k   & 458    & $>$10k  & $>$4GB \\
 \texttt{amba08c}   & 19 & 16  & 56  &  511  && 667    & $>$10k & $>$10k & $>$10k & $>$10k & $>$10k & $>$10k & $>$10k & $>$10k  & $>$4GB \\
 \texttt{amba09c}   & 21 & 18  & 61  &  583  && 9.3k   & $>$10k & $>$10k & $>$10k & 2.8k   & $>$10k & 3.9k   & 1.0k   & $>$10k  & $>$4GB \\
 \texttt{amba10c}   & 23 & 19  & 65  &  630  && 271    & $>$10k & $>$10k & $>$10k & 3.4k   & $>$10k & 3.6k   & 1.6k   & $>$10k  & $>$4GB \\
\cmidrule{1-5}
\cmidrule{7-16}
 \texttt{amba02b}   & 7  & 8   & 31  &  189  && 1.7    & 1.2k   & 8.3k   & 23     & 24     & 57     & 20     & 15     & $>$10k  & $>$4GB \\
 \texttt{amba03b}   & 9  & 10  & 36  &  231  && 13     & $>$10k & $>$10k & 207    & 70     & 131    & 29     & 30     & $>$10k  & $>$4GB \\
 \texttt{amba04b}   & 11 & 11  & 42  &  286  && 84     & $>$10k & $>$10k & 3.8k   & 761    & 4.2k   & 831    & 504    & $>$10k  & $>$4GB \\
 \texttt{amba05b}   & 13 & 13  & 47  &  344  && 403    & $>$10k & $>$10k & 2.5k   & 278    & 2.0k   & 210    & 216    & $>$10k  & $>$4GB \\
 \texttt{amba06b}   & 15 & 14  & 52  &  391  && 903    & $>$10k & $>$10k & 3.2k   & 394    & 7.8k   & 366    & 209    & $>$10k  & $>$4GB \\
 \texttt{amba07b}   & 17 & 15  & 57  &  438  && 1.5k   & $>$10k & $>$10k & 7.3k   & 1.1k   & $>$10k & 502    & 634    & $>$10k  & $>$4GB \\
 \texttt{amba08b}   & 19 & 16  & 62  &  486  && $>$10k & $>$10k & $>$10k & $>$10k & $>$10k & $>$10k & $>$10k & $>$10k & $>$10k  & $>$4GB \\
 \texttt{amba09b}   & 21 & 18  & 68  &  558  && $>$10k & $>$10k & $>$10k & $>$10k & 4.9k   & $>$10k & 1.2k   & 1.9k   & $>$10k  & $>$4GB \\
 \texttt{amba10b}   & 23 & 19  & 73  &  606  && 3.7k   & $>$10k & $>$10k & $>$10k & 6.4k   & $>$10k & 3.0k   & 5.8k   & $>$10k  & $>$4GB \\
\bottomrule[1.3pt]
\end{tabular}
\end{table*}

\end{document}